\documentclass[10pt]{article}

\usepackage{fullpage}

\usepackage{fancyhdr}
\usepackage{ifthen,float}
\usepackage{amsmath}
\usepackage{amsfonts}
\usepackage{color,xcolor}
\usepackage{hyperref} \hypersetup{colorlinks=true, breaklinks=true, urlcolor=black, linkcolor=black, citecolor=black}
\usepackage{amssymb}
\usepackage{enumerate}
\usepackage{subcaption}
\usepackage[clock]{ifsym}

\usepackage{tikz}
\usetikzlibrary{patterns}
\pgfdeclarepatternformonly{stripes}{\pgfqpoint{-4pt}{-4pt}}{\pgfqpoint{8pt}{8pt}}{\pgfqpoint{4pt}{4pt}}%
{
  \pgfsetlinewidth{1.414pt}
  \pgfpathmoveto{\pgfqpoint{0pt}{0pt}}
  \pgfpathlineto{\pgfqpoint{4pt}{4pt}}
  \pgfusepath{stroke}
}

\colorlet{colorA}{black!20!red}
\colorlet{colorB}{black!20!blue}

\usepackage[boxruled,vlined,linesnumbered]{algorithm2e}
\SetKwBlock{SubAlgoBlock}{}{end}
\newcommand{\SubAlgo}[2]{#1 \SubAlgoBlock{#2}}
\let\oldnl\nl
\newcommand{\nonl}{\renewcommand{\nl}{\let\nl\oldnl}}

\SetKw{Var}{var}
\SetKw{Op}{operation}
\SetKw{RecA}{when a message}
\SetKw{RecB}{is received from}
\SetKw{Broadcast}{FIFO broadcast}
\SetKw{Wait}{wait until}

\newcommand{\DataFont}[1]{\textsf{#1}}
\newcommand{\FunctionFont}[1]{\texttt{#1}}
    
\newcommand{\VarNameX}{X}          \SetKwData{KwX}{\VarNameX}   \newcommand{\AlgoX}{\DataFont{\VarNameX}}
\newcommand{\VarNameVC}{ValClock}  \SetKwData{KwVC}{\VarNameVC} \newcommand{\AlgoVC}{\DataFont{\VarNameVC}}
\newcommand{\VarNameSC}{SendClock} \SetKwData{KwSC}{\VarNameSC} \newcommand{\AlgoSC}{\DataFont{\VarNameSC}}
\newcommand{\VarNameG}{G}          \SetKwData{KwG}{\VarNameG}   \newcommand{\AlgoG}{\DataFont{\VarNameG}}
\newcommand{\VarNameV}{V}          \SetKwData{KwV}{\VarNameV}   \newcommand{\AlgoV}{\DataFont{\VarNameV}}

\newcommand{\VarNamev}{v}   \SetKwData{Kwv}{\VarNamev}   \SetKwFunction{KwGV}{\VarNamev} \newcommand{\Algov}{\DataFont{\VarNamev}} \newcommand{\AlgoGV}{\FunctionFont{\VarNamev}}
\newcommand{\VarNamek}{k}   \SetKwData{Kwk}{\VarNamek}   \SetKwFunction{KwGK}{\VarNamek} \newcommand{\Algok}{\DataFont{\VarNamek}} \newcommand{\AlgoGK}{\FunctionFont{\VarNamek}}
\newcommand{\VarNamet}{t}   \SetKwData{Kwt}{\VarNamet}   \SetKwFunction{KwGT}{\VarNamet} \newcommand{\Algot}{\DataFont{\VarNamet}} \newcommand{\AlgoGT}{\FunctionFont{\VarNamet}}
\newcommand{\VarNamecl}{cl} \SetKwData{Kwcl}{\VarNamecl} \SetKwFunction{KwGCL}{\VarNamecl} \newcommand{\Algocl}{\DataFont{\VarNamecl}} \newcommand{\AlgoGCL}{\FunctionFont{\VarNamecl}}

\newcommand{\VarNameWrite}{update}   \SetKwFunction{KwWrite}{\VarNameWrite} \newcommand{\AlgoWrite}{\FunctionFont{\VarNameWrite}}
\newcommand{\VarNameSnap}{snapshot} \SetKwFunction{KwSnap}{\VarNameSnap}  \newcommand{\AlgoSnap}{\FunctionFont{\VarNameSnap}}
\newcommand{\VarNameM}{message}     \SetKwFunction{KwM}{\VarNameM}     \newcommand{\AlgoM}{\FunctionFont{\VarNameM}}
    
\newcommand{\Algog}{g}
\newcommand{\REG}{\DataFont{REG}}

\makeatletter
\renewcommand{\algocf@caption@boxruled}{%
  \hrule
  \hbox to \hsize{%
    \vrule\hskip-0.4pt
    \vbox{   
       \vskip\interspacetitleboxruled%
       \unhbox\algocf@capbox\hfill
       \vskip\interspacetitleboxruled
       }%
     \hskip-0.4pt\vrule%
   }\nointerlineskip%
}%
\makeatother

\usepackage{amsthm}

\theoremstyle {plain}
\newtheorem{theorem}{Théorème}
\newtheorem{proposition}[theorem]{Proposition}

\newtheorem{lemma}[theorem]{Lemme}

\theoremstyle {definition}
\newtheorem{definition}{Définition}

\begin{document}

\cfoot{\thepage}

\title{On Composition and Implementation of Sequential Consistency (Extended Version)}
\author{Matthieu Perrin \\ LINA -- University of Nantes \\ \url{matthieu.perrin@univ-nantes.fr}
  \and Matoula Petrolia \\ LINA -- University of Nantes \\ \url{stamatina.petrolia@univ-nantes.fr}
  \and Achour Mostéfaoui \\ LINA -- University of Nantes \\ \url{achour.mostefaoui@univ-nantes.fr}
  \and Claude Jard \\ LINA -- University of Nantes \\ \url{claude.jard@univ-nantes.fr}
}
\date{}

\sloppy

\maketitle

\begin{abstract}
It has been proved that to implement a linearizable shared memory in synchronous message-passing systems it is necessary to wait for a time proportional to the uncertainty in the latency of the network for both read and write operations, while waiting during read or during write operations is sufficient for sequential consistency. 

This paper extends this result to crash-prone asynchronous systems. We propose a distributed algorithm that builds a sequentially consistent shared memory abstraction with snapshot on top of an asynchronous message-passing system where less than half of the processes may crash. We prove that it is only necessary to wait when a read/snapshot is immediately preceded by a write on the same process. 

We also show that sequential consistency is composable in some cases commonly encountered: 1) objects that would be linearizable if they were implemented on top of a linearizable memory become sequentially consistent when implemented on top of a sequential memory while remaining composable and 2) in round-based algorithms, where each object is only accessed within one round.

\vspace{2mm}
\begin{center}
\textbf{Key words}
\end{center}

Asynchronous message-passing system, Crash-failures, Composability, Sequential consistency, Shared memory, Snapshot.
\end{abstract}

%%%%%%%%%%%%%%%%%%%%
\section{Introduction}
%%%%%%%%%%%%%%%%%%%%

A distributed system is abstracted as a set of entities (nodes, processes, agents, etc) that communicate with each other using a communication medium. 
The two most used communication media are communication channels (message-passing system) and shared memory (read/write operations). 
Programming  with shared objects is generally more convenient as it offers a higher level of abstraction to the programmer,
therefore facilitates the work of designing distributed applications. 
A natural question is the level of consistency ensured by shared objects. 
An intuitive property is that shared objects should behave as if all processes accessed the same physical copy of the object. 
\emph{Sequential consistency}~\cite{lamport1979make} ensures that all the operations that happen in a distributed history 
appear as if they were executed sequentially, in an order that respects the sequential order of each process (called the \emph{process order}). 

Unfortunately, sequential consistency is not composable: if a program uses two or more objects,
despite each object being sequentially consistent individually, the set of all objects may not be sequentially consistent.
An example is shown in Fig.~\ref{fig:composability}, where two processes share two registers named $X$ and $Y$; although the operations of each register
may be totally ordered (the read precedes the write), it is impossible to order all the operations at once.
\emph{Linearizability}~\cite{herlihy1990linearizability} overcomes this limitation by adding constraints on real time: each operation appears at a single 
point in time, between its start event and its end event. As a consequence, linearizability enjoys the locality property \cite{herlihy1990linearizability} that ensures its composability. 
Because of this composability, much more effort has been focused on linearizability than on sequential consistency so far. 
However, one of our contributions implies that in asynchronous systems where no global clock can be implemented to measure real time, a process cannot distinguish between linearizability and sequential consistency, thus 
the connection to real time
seems to be a worthless --- though costly --- guarantee.

\begin{figure}[t]
    \centering
    \begin{tikzpicture}

      \draw[->] (-0.3,0.5) node[left]{$p_1$} -- (4.4,0.5) ;
      \draw[->] (-0.3,0.0) node[left]{$p_0$} -- (4.4,0.0) ;

      \draw[->,very thick, colorA] (4.4,-0.5) -- (-0.3,1) node[left]{$l_X$};
      \draw[->,very thick, colorB] (4.4,1) -- (-0.3,-0.5) node[left]{$l_Y$};

      \draw[fill=colorA!10] (1.0,0.5) +(-0.8,-0.2) rectangle +(0.8,0.2) +(0,0) node{\footnotesize $X.\textsf{write}(1)$};
      \draw[fill=colorB!10] (3.0,0.5) +(-0.9,-0.2) rectangle +(0.9,0.2) +(0,0) node{\footnotesize $Y.\textsf{read}\rightarrow 0$};
      \draw[fill=colorB!10] (1.0,0) +(-0.8,-0.2) rectangle +(0.8,0.2) +(0,0) node{\footnotesize $Y.\textsf{write}(1)$};
      \draw[fill=colorA!10] (3.0,0) +(-0.9,-0.2) rectangle +(0.9,0.2) +(0,0) node{\footnotesize $X.\textsf{read}\rightarrow 0$};
    \end{tikzpicture}
    \caption{Sequential consistency is not composable: registers $X$ and $Y$ are both sequentially consistent but their composition is not.}
    \label{fig:composability}
\end{figure}
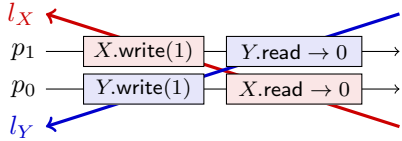

In this paper we focus on message-passing distributed systems. In such systems a shared memory is not a physical object; it has to be built using the underlying message-passing communication network. 
Several bounds have been found on the cost of sequential consistency and linearizability in synchronous distributed systems, where the transit time for any message is in a range $[d-u,d]$,
where $d$ and $u$ are called respectively the \emph{latency} and the \emph{uncertainty} of the network. Let us consider an implementation of a shared memory, and let $r$ (resp. $w$) be the 
worst case latency of any read (resp. write) operation. Lipton and Sandberg proved in \cite{lipton1988pram} that, if the algorithm implements a sequentially consistent memory, the inequality $r+w\geq d$ must hold.
Attiya and Welch refined this result in~\cite{attiya1994sequential}, proving that each kind of operations could have a 0-latency implementation for sequential consistency (though not both in the same implementation)
but that the time duration of both kinds of operations has to be at least linear in $u$ in order to ensure linearizability.

Therefore the following questions arise. Are there applications for which the lack of composability of sequential consistency is not a problem? For these applications, can we expect the same benefits in weaker message-passing models, such as asynchronous failure-prone systems, from using sequentially consistent objects rather than linearizable objects?

To illustrate the contributions of the paper, we also address a higher level operation: a snapshot operation \cite{Afek93} that allows to read in a single operation a whole set of registers. A sequentially consistent snapshot is such that the set of values it returns may be returned by a sequential execution. This operation is very useful as it has been proved \cite{Afek93} that linearizable snapshots can be wait-free implemented from single-writer/multi-reader registers. Indeed, assuming a snapshot operation does not bring any additional power with respect to shared registers. Of course this induces an additional cost: the best known simulation needs $O(n\log n)$ basic read/write operations to implement each of the snapshot operations and the associated update operation \cite{AttiyaR98}. Such an operation brings a programming comfort as it reduces the ``noise'' introduced by asynchrony and failures \cite{G98} and is particularly used in round-based computations \cite{Gafni98} we consider for the study of the composability of sequential consistency.

\paragraph{Contributions.} This paper has three major contributions. (1) It identifies two contexts that can benefit from the use of sequential consistency: round-based algorithms that use a different shared object for each round, and asynchronous shared-memory
systems, where programs can not differentiate a sequentially consistent memory from a linearizable memory. 
(2) It proposes an implementation of a sequentially consistent memory where waiting is only required when a write is immediately followed by a read. 
This extends the result presented in~\cite{attiya1994sequential}, which only applies to synchronous failure-free systems, to failure-prone asynchronous systems. 
(3) The proposed algorithm also implements a sequentially consistent snapshot operation the cost of which compares very favorably with the best existing linearizable implementation to our knowledge (the stacking of the snapshot algorithm of Attiya and Rachman \cite{AttiyaR98} over the ABD simulation of linearizable registers).

\paragraph{Outline.}
The remainder of this article is organized as follows. In Section~\ref{sec:composition}, we define more formally sequential consistency,
and we present special contexts in which it becomes composable. Then, in Section~\ref{sec:implementation}, we present our implementation 
of shared memory and study its complexity. Finally, Section~\ref{sec:conclusion} concludes the paper.

%%%%%%%%%%%%%%%%%%%%
\section{Sequential Consistency and Composability}\label{sec:composition}
%%%%%%%%%%%%%%%%%%%%

%%%%%%%%%%%%%%%%%%%%
\subsection{Definitions}
%%%%%%%%%%%%%%%%%%%%

In this section we recall the definitions of the most important notions we discuss in this paper: 
two consistency criteria, sequential consistency ($SC$, Def.~\ref{def:SC}, \cite{lamport1979make}) and linearizability ($L$, Def.~\ref{def:Lin}, \cite{herlihy1990linearizability}), 
as well as composability (Def.~\ref{def:comp}). A consistency criterion associates a set of admitted \emph{histories}
to the \emph{sequential specification} of each given object. A history is a representation of an execution. It contains a set of operations, 
that are partially ordered according to the sequential order of each process, called \emph{process order}. 
A sequential specification 
is a language, i.e. a set of sequential (finite and infinite) words. For a consistency criterion $C$ and a sequential specification $T$, 
we say that an algorithm implements a $C(T)$-consistent object if all its executions can be modelled by a history that belongs to $C(T)$, 
that contains all returned operations and only invoked operations. 
Note that this implies that if a process crashes during an operation, then the operation will appear in the history as if it was complete or as if it never took place at all.

\begin{definition}[Linear extension]\label{def:lin_order}
Let $H$ be a history and $T$ be a sequential specification. A \emph{linear extension} $\le$ is a total order on all the operations of $H$, that contains the process order, and such that each event $e$ has a finite past $\{e' : e'\le e\}$ according to the total order. 
\end{definition}

\begin{definition}[Sequential Consistency]\label{def:SC}
    Let $H$ be a history and $T$ be a sequential specification. The history $H$ is \emph{sequentially consistent} regarding $T$, denoted $H\in SC(T)$, if there exists a linear extension $\le$ such that the word composed of all the operations of $H$ ordered by $\le$ belongs to $T$. 
\end{definition}

\begin{definition}[Linearizability]\label{def:Lin}
    Let $H$ be a history and $T$ be a sequential specification. The history $H$ is \emph{linearizable} regarding $T$, 
    denoted $H\in L(T)$, if there exists a linear extension $\le$ such that 
    (1) for two operations $a$ and $b$, if the end of $a$ precedes the beginning of $b$ in real time, then $a\le b$ and 
    (2) the word formed of all the operations of $H$ ordered by $\le$ belongs to $T$. 
\end{definition}

Let $T_1$ and $T_2$ be two sequential specifications. We define the \emph{composition} of $T_1$ and $T_2$, denoted by $T_1\times T_2$, 
as the set of all the interleaved sequences of a word from $T_1$ and a word from $T_2$. An interleaved sequence of two words $l_1$ and $l_2$ is
a word composed of the disjoint union of all the letters of $l_1$ and $l_2$, that appear in the same order as they appear in $l_1$ and $l_2$. 
For example, the words $ab$ and $cd$ have six interleaved sequences: $abcd$, $acbd$, $acdb$, $cabd$, $cadb$ and $cdab$. 

A consistency criterion $C$ is composable (Def.~\ref{def:comp}) if the composition of a $C(T_1)$-consistent object and a $C(T_2)$-consistent object
is a $C(T_1\times T_2)$-consistent object. Linearizability is composable, and sequential consistency is not. 

\begin{definition}[Composability]\label{def:comp}
For a history $H$ and a sequential specification $T$, let us denote by $H_{T}$ the sub-history of $H$ that only contains the operations belonging to $T$. 

A consistency criterion $C$ is \emph{composable} if, for all sequential specifications $T_1$ and $T_2$ and all histories $H$ containing only events 
on $T_1$ and $T_2$, $(H_{T_1} \in C(T_1) \text{ and } H_{T_2} \in C(T_2))$ imply $H \in C(T_1\times T_2)$.
\end{definition}

\subsection{From Linearizability to Sequential Consistency}

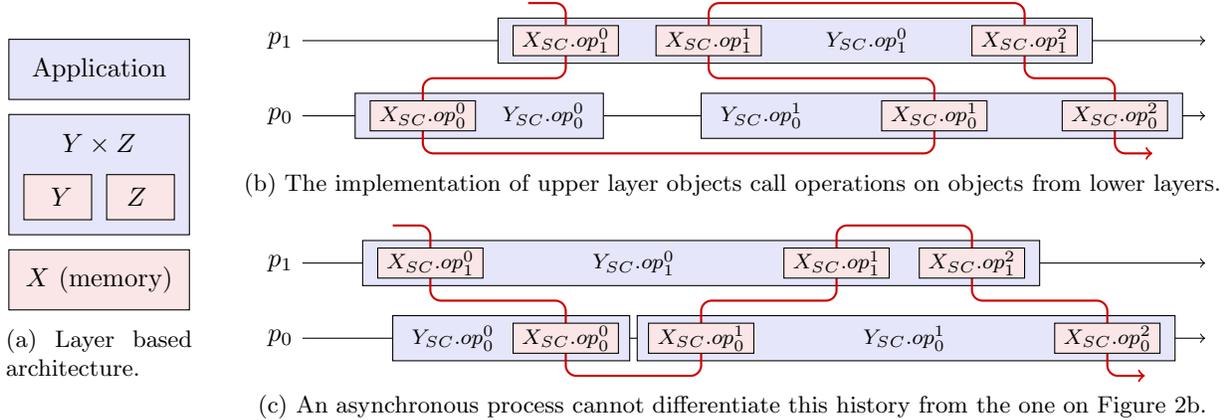
\begin{figure}[t]
  \begin{subfigure}{0.15\textwidth}
    \centering
    \begin{tikzpicture}
      
      \draw[fill=colorB!10] (0.3,2.8) rectangle (2.7,3.6) ;
      \draw[fill=colorB!10] (0.3,1) rectangle (2.7,2.6) ;
      \draw[fill=colorA!10] (0.5,1.2) rectangle (1.4,1.8) ;
      \draw[fill=colorA!10] (1.6,1.2) rectangle (2.5,1.8) ;
      \draw[fill=colorA!10] (0.3,0) rectangle (2.7,0.8) ;

      \draw (1.5,3.2) node{Application};
      \draw (1.5,2.2) node{$Y \times Z$};
      \draw (1,1.5) node{$Y$};
      \draw (2,1.5) node{$Z$};
      \draw (1.5,0.4) node{$X$ (memory)};

    \end{tikzpicture}
  \caption{Layer based architecture.}
  \label{fig:lin:archi}
  \end{subfigure}
  \hspace{\fill}
  \begin{minipage}{0.83\textwidth} 
  \begin{subfigure}{\textwidth}
    \centering
    \begin{tikzpicture}
      
      \draw[->] (-0.5,1) node[left]{$p_1$} -- (11.5,1);
      \draw[->] (-0.5,0) node[left]{$p_0$} -- (11.5,0);

      \draw[fill=colorB!10] (0,1) +(2.1,-0.3) rectangle +(10,0.3) +(7,0) node{\footnotesize $Y_{SC}.op_1^0$};

      \draw[fill=colorB!10] (0,0) +(0.2,-0.3) rectangle +(3.5,0.3) +(2.7,0) node{\footnotesize $Y_{SC}.op_0^0$};
      \draw[fill=colorB!10] (0,0) +(4.8,-0.3) rectangle +(11.2,0.3) +(5.6,0) node{\footnotesize $Y_{SC}.op_0^1$};

      \draw[->,colorA,rounded corners, thick]
      (2.5,1.5) -- (3,1.5) -- (3,0.5) --
      (1.1,0.5) -- (1.1,-0.5) --
      (7.9,-0.5) -- (7.9,0.5) --
      (4.9,0.5)  -- (4.9,1.5) --
      (9.1,1.5) -- (9.1,0.5)  --
      (10.3,0.5) -- (10.3,-0.5) -- (10.8,-0.5);

      \draw[fill=colorA!10] (3.0,1) +(-0.7,-0.2) rectangle +(0.7,0.2) +(0,0) node{\footnotesize $X_{SC}.op_1^0$};
      \draw[fill=colorA!10] (4.9,1) +(-0.7,-0.2) rectangle +(0.7,0.2) +(0,0) node{\footnotesize $X_{SC}.op_1^1$};
      \draw[fill=colorA!10] (9.1,1) +(-0.7,-0.2) rectangle +(0.7,0.2) +(0,0) node{\footnotesize $X_{SC}.op_1^2$};

      \draw[fill=colorA!10] (1.1,0) +(-0.7,-0.2) rectangle +(0.7,0.2) +(0,0) node{\footnotesize $X_{SC}.op_0^0$};
      \draw[fill=colorA!10] (7.9,0) +(-0.7,-0.2) rectangle +(0.7,0.2) +(0,0) node{\footnotesize $X_{SC}.op_0^1$};
      \draw[fill=colorA!10] (10.3,0) +(-0.7,-0.2) rectangle +(0.7,0.2) +(0,0) node{\footnotesize $X_{SC}.op_0^2$};
      
    \end{tikzpicture}
  \caption{The implementation of upper layer objects call operations on objects from lower layers.}
  \label{fig:lin:histSC}
  \end{subfigure}

  \vspace{3mm}

  \begin{subfigure}{\textwidth}
    \centering
    \begin{tikzpicture}
      
      \draw[->] (-0.5,1) node[left]{$p_1$} -- (11.5,1);
      \draw[->] (-0.5,0) node[left]{$p_0$} -- (11.5,0);

      \draw[fill=colorB!10] (0,1) +(0.3,-0.3) rectangle +(9.3,0.3) +(3.9,0) node{\footnotesize $Y_{SC}.op_1^0$};

      \draw[fill=colorB!10] (0,0) +(0.7,-0.3) rectangle +(3.85,0.3) +(1.5,0) node{\footnotesize $Y_{SC}.op_0^0$};
      \draw[fill=colorB!10] (0,0) +(3.95,-0.3) rectangle +(11.1,0.3) +(7.5,0) node{\footnotesize $Y_{SC}.op_0^1$};

      \draw[->,colorA,rounded corners, thick]
      (0.7,1.5) -- (1.2,1.5) -- (1.2,0.5) --
      (3.0,0.5) -- (3.0,-0.5) --
      (4.8,-0.5) -- (4.8,0.5) --
      (6.6,0.5)  -- (6.6,1.5) --
      (8.4,1.5) -- (8.4,0.5)  --
      (10.2,0.5) -- (10.2,-0.5) -- (10.7,-0.5);

      \draw[fill=colorA!10] (1.2,1) +(-0.7,-0.2) rectangle +(0.7,0.2) +(0,0) node{\footnotesize $X_{SC}.op_1^0$};
      \draw[fill=colorA!10] (6.6,1) +(-0.7,-0.2) rectangle +(0.7,0.2) +(0,0) node{\footnotesize $X_{SC}.op_1^1$};
      \draw[fill=colorA!10] (8.4,1) +(-0.7,-0.2) rectangle +(0.7,0.2) +(0,0) node{\footnotesize $X_{SC}.op_1^2$};

      \draw[fill=colorA!10] (3.0,0) +(-0.7,-0.2) rectangle +(0.7,0.2) +(0,0) node{\footnotesize $X_{SC}.op_0^0$};
      \draw[fill=colorA!10] (4.8,0) +(-0.7,-0.2) rectangle +(0.7,0.2) +(0,0) node{\footnotesize $X_{SC}.op_0^1$};
      \draw[fill=colorA!10] (10.2,0) +(-0.7,-0.2) rectangle +(0.7,0.2) +(0,0) node{\footnotesize $X_{SC}.op_0^2$};
      
    \end{tikzpicture}
  \caption{An asynchronous process cannot differentiate this history from the one on Figure~\ref{fig:lin:histSC}.}
  \label{fig:lin:histL}
  \end{subfigure}
  \end{minipage}
  \caption{In layer based program architecture running on asynchronous systems, local clocks of different processes can be distorted such that it 
  is impossible to differentiate a sequentially consistent execution from a linearizable execution.}
  \label{fig:lin}
\end{figure}

Software developers usually abstract the complexity of their system gradually, which results in a layered software architecture: 
at the top level, an application is built on top of several objects specific to the application, 
themselves built on top of lower levels. Such an architecture is represented in Fig.~\ref{fig:lin:archi}.
The lowest layer usually consists of one or several objects provided by the system itself, typically a shared memory. 
The system can ensure sequential consistency globally on all the provided objects, therefore composability is not required for this level.
Proposition~\ref{prop:lin_SC} expresses the fact that, in asynchronous systems, replacing a linearizable object by a sequentially consistent one 
does not affect the correctness of the programs running on it circumventing the non composability of sequential consistency.
This result may have an impact on parallel architectures, such as modern multi-core processors and, to a higher extent, 
high performance supercomputers, for which the communication with a linearizable central shared memory is very costly, 
and weak memory models such as cache consistency~\cite{goodman1991cache} make the writing of programs tough. 

\begin{proposition}\label{prop:lin_SC}
Let $A$ be an algorithm that implements an $SC(Y)$-consistent object when it is executed on an asynchronous system providing an $L(X)$-consistent object. Then $A$ also implements an $SC(Y)$-consistent object when it is executed in an asynchronous system providing an $SC(X)$-consistent object.
\end{proposition}

\begin{proof}
  Let $A$ be an algorithm that implements an $SC(Y)$-consistent object when it is executed on 
  an asynchronous system providing an $L(X)$-consistent object. 
  
  Let us consider a history $H_{SC}$ obtained by the execution of $A$ in an asynchronous system providing a $SC(X)$-consistent object. 
  Such a history is depicted on Fig.~\ref{fig:lin:histSC}. The history $H_{SC}$ contains operations on $X$ (in red in Fig.~\ref{fig:lin:histSC}), 
  as well as on $Y$ (in blue in Fig.~\ref{fig:lin:histSC}). 

  We will now build another history $H_{L}$, in which the operations on $X$ are linearizable, and the operations on $Y$ consist in the same calls to operations on $X$. 
  Such a history is depicted on Fig.~\ref{fig:lin:histL}. The only difference between the histories on Fig.~\ref{fig:lin:histSC} and~\ref{fig:lin:histL} is the
  way the two processes experience time. As the system is asynchronous, it is impossible for them to distinguish them. 
  
  Let us enumerate all the operations made on $X$ in their linear extension $\le$ required for sequential consistency. 
  Now, we build the execution $H_L$ in which the $i^{\text{th}}$ operation on $X$ of $H_{SC}$ is called on an $L(X)$-consistent object
  at time $2i$ seconds and lasts for one second. As no two operations overlap, and the operations happen in the
  same order in $H_L$ and in the linearization of $H_{SC}$, $\le$ is the only linear extension accepted by linearizability.
  Therefore, all operations can return the same values in $H_L$ and in $H_{SC}$ (and they will if $X$ is deterministic). 
  Now let us assume all operations on $X$ in $H_L$ were called by algorithm $A$, in the same pattern as in $H_{SC}$. 
  When considering the operations on $Y$, $H_L$ is $SC(Y)$-consistent. Moreover, as $A$ works on asynchronous systems
  and the same values were returned by $X$ in $H_{SC}$ and in $H_{L}$, $A$ returns the same values in both histories. 
  Therefore, $H_{SC}$ is also $SC(Y)$-consistent. 
\end{proof}

An interesting point about Proposition~\ref{prop:lin_SC} is that it allows sequentially consistent --- but not linearizable ---
objects to be composable. Let $A_Y$ and $A_Z$ be two algorithms that implement $L(Y)$-consistent and $L(Z)$-consistent objects
when they are executed on an asynchronous system providing an $L(X)$-consistent object, like on Fig.~\ref{fig:lin:archi}.
As linearizability is stronger than sequential consistency, according to Proposition~\ref{prop:lin_SC}, executing $A_Y$ and $A_Z$
on an asynchronous system providing an $SC(X)$-consistent object would implement sequentially consistent --- yet not linearizable --- 
objects. However, in a system providing the linearizable object $X$, by composability of linearizability, 
the composition of $A_Y$ and $A_Z$ implements an $L(Y\times Z)$-consistent object. Therefore, by Proposition~\ref{prop:lin_SC} again, 
in a system providing the sequentially consistent object $X$, the composition also implements an $SC(Y\times Z)$-consistent object.
In this example, the sequentially consistent versions of $Y$ and $Z$ derive their composability from an anchor to a \emph{common time},
given by the sequentially consistent memory, that can differ from \emph{real time}, required by linearizability.

%%%%%%%%%%%%%%%%%%%%
\subsection{Round-Based Computations}\label{sec:round}
%%%%%%%%%%%%%%%%%%%%

Even at a single layer, a program can use several objects that are not composable, but that are used in a fashion so that the non-composability is invisible to the program. Let us illustrate this with round-based algorithms. The synchronous distributed computing model has been extensively studied and well-understood leading the researchers to try to offer the same comfort when dealing with asynchronous systems, hence the introduction of synchronizers \cite{Awerbuch85}. A synchronizer slices a computation into phases during which each process executes three steps: send/write, receive/read and then local computation. This model has been extended to failure prone systems in the round-by-round computing model \cite{Gafni98} and to the Heard-Of model \cite{CS09} among others. Such a model is particularly interesting when the termination of a given program is only eventual. Indeed, some problems are undecidable in failure prone purely asynchronous systems. In order to circumvent this impossibility, eventually or partially synchronous systems have been introduced \cite{DLS88}. In such systems the termination may hold only after some finite but unbounded time, and the algorithms are implemented by the means of a series of asynchronous rounds each using its own shared objects.

In the round-based computing model, the execution is sliced into a sequence of asynchronous rounds. During each round, a new data structure (usually a single-writer/multi-reader register per process) is created and it is the only shared object used to communicate during the round. 
At the end of the round, each process destroys its local accessor to the object, so that it can no more access it. Note that the rounds are asynchronous:
the different processes do not necessarily start and finish their rounds at the same time. Moreover, a process may not terminate a round, and keep accessing the same shared object forever or may crash during this round and stop executing. A round-based execution is illustrated in Fig.~\ref{fig:rounds:hist}.

In Proposition~\ref{prop:round}, we prove that sequentially consistent objects of different rounds behave well together:
as the ordering added between the operations of two different objects always follows the round numbering, that is 
consistent with the program order already contained in the linear extension of each object, the composition of 
all these objects cannot create loops (Figure~\ref{fig:rounds:hist}). 
Putting together this result and Proposition~\ref{prop:lin_SC}, all the algorithms 
that use a round-based computation model can benefit of any improvement on the implementation of an array of 
single-writer/multi-reader register that sacrifices linearizability for sequential consistency. Note that this remains true whatever is the data structure used during each round. The only constraint is that a sequentially consistent shared data structure can be accessed during a unique round. If each object is sequentially consistent then the whole execution is consistent. 

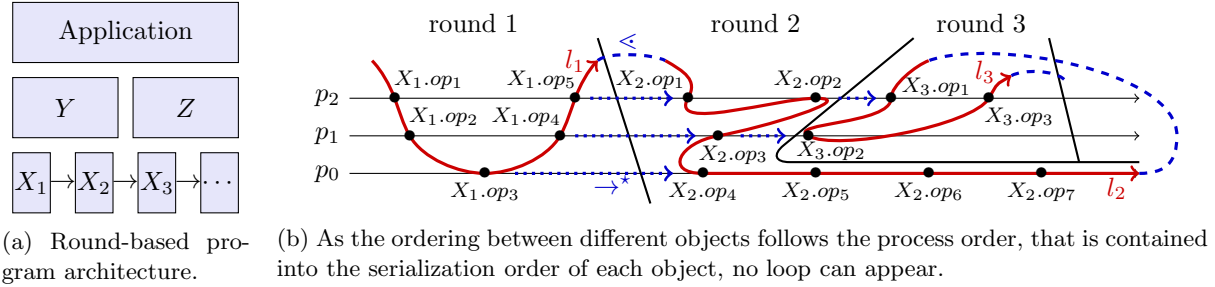
\begin{figure}[t]
  \hspace{\fill}
  \begin{subfigure}{0.2\textwidth}
    \centering
    \begin{tikzpicture}
      
      \draw[fill=colorB!10] (0,2) rectangle (3,2.8) ;
      \draw[fill=colorB!10] (0,1) rectangle (1.4,1.8) ;
      \draw[fill=colorB!10] (1.6,1) rectangle (3,1.8) ;
      \draw[fill=colorB!10] (0,0) rectangle (0.5,0.8) ;
      \draw[fill=colorB!10] (0.8333,0) rectangle (1.3333,0.8) ;
      \draw[fill=colorB!10] (1.666,0) rectangle (2.166,0.8) ;
      \draw[fill=colorB!10] (2.5,0) rectangle (3,0.8) ;

      \draw (1.5,2.4) node{Application};
      \draw (0.7,1.4) node{$Y$};
      \draw (2.3,1.4) node{$Z$};
      \draw (0.25,0.4) node{$X_1$};
      \draw (0.6666,0.4) node{$\rightarrow$};
      \draw (1.08333,0.4) node{$X_2$};
      \draw (1.5,0.4) node{$\rightarrow$};
      \draw (1.917,0.4) node{$X_3$};
      \draw (2.33333,0.4) node{$\rightarrow$};
      \draw (2.75,0.4) node{$\cdots$};

    \end{tikzpicture}
  \caption{Round-based program architecture.}
  \label{fig:rounds:archi}
  \end{subfigure}
  \hspace{\fill}
  \begin{subfigure}{0.75\textwidth}
    \centering
    \begin{tikzpicture}

      \draw (0.15,2) node{round 1};
      \draw (3.9 ,2) node{round 2};
      \draw (6.9 ,2) node{round 3};

      \draw[thick,rounded corners=10] (1.8,1.8) -- (2.5,-0.4);
      \draw[thick,rounded corners=10] (6,1.8) -- (4,0.15) -- (9,0.15) ;
      \draw[thick,rounded corners=10] (7.8,1.8) -- (8.2,0.15);

      \draw[->] (-1.5,1.0) node[left]{$p_2$} -- (9,1.0) ;
      \draw[->] (-1.5,0.5) node[left]{$p_1$} -- (9,0.5) ;
      \draw[->] (-1.5,0.0) node[left]{$p_0$} -- (9,0.0) ;

      \draw[->,colorA,very thick]
      (-1.2,1.5) to[out=-45,in=110,distance=5]
      (-0.9,1.0) to[out=-70,in=120,distance=5]
      (-0.7,0.5) to[out=-60,in=180,distance=10]
      (0.3,0.0)  to[out=0 ,in=-120,distance=10]
      (1.3,0.5)  to[out=60,in=-110,distance=5]
      (1.5,1.0)  to[out=70,in=-135,distance=5]
      (1.8,1.5);

      \draw[colorB,very thick,dashed]
      (1.8,1.5) to[out=45,in=150,distance=5]
      (2.7,1.5);

      \draw[->,colorA,very thick]
      (2.7,1.5) to[out=-30,in=60,distance=8]
      (3.0,1.0) to[out=-120,in=180,distance=15]
      (4.7,1.0)  to[out=0 ,in=10,distance=17]
      (3.4,0.5)  to[out=-170,in=180,distance=15]
      (3.2,0.0) to[out=0,in=180,distance=20]
      (4.2,0.0) -- (8.2,0.0) -- (9,0.0);

      \draw[colorB,very thick,dashed]
      (9,0) to[out=0,in=-90,distance=10]
      (9.5,.5) to[out=90,in=20,distance=25]
      (6.2,1.5);

      \draw[->,colorA,very thick]
      (6.2,1.5) to[out=-160,in=90,distance=5]
      (5.7,1.0) to[out=-90,in=150,distance=10]
      (4.6,0.5) to[out=-30,in=-120,distance=12]
      (7.0,1.0)  to[out=60,in=-150,distance=7]
      (7.3,1.3);
      
      \draw[colorB,very thick,dashed]
      (7.3,1.3) to[out=30,in=135,distance=5]
      (8,1.25);

      \draw[->, colorB,very thick,dotted] (1.7,1.0) -- (2.8,1.0);
      \draw[->, colorB,very thick,dotted] (1.5,0.5) -- (3.1,0.5);
      \draw[->, colorB,very thick,dotted] (0.7,0.0) -- (2.8,0.0);

      \draw[->, colorB,very thick,dotted] (5,1.0) -- (5.5,1.0);
      \draw[->, colorB,very thick,dotted] (3.7,0.5) -- (4.3,0.5);

      \draw[colorA] (1.5,1.5)  node{$l_1$};
      \draw[colorA] (8.7,-0.2)  node{$l_2$};
      \draw[colorA] (6.95,1.35)  node{$l_3$};
      \draw[colorB] (2.2,1.75)  node{$\lessdot$};
      \draw[colorB] (2,-0.15)  node{$\rightarrow^\star$};

      \draw (-0.9,1.0)node{$\bullet$} +(0.45,0) node[above]{\footnotesize $X_1.op_1$};
      \draw (-0.7,0.5)node{$\bullet$} +(0.45,0) node[above]{\footnotesize $X_1.op_2$};
      \draw (0.3,0.0) node{$\bullet$} node[below]{\footnotesize $X_1.op_3$};
      \draw (1.3,0.5) node{$\bullet$} +(-0.45,0) node[above]{\footnotesize $X_1.op_4$};
      \draw (1.5,1.0) node{$\bullet$} +(-0.45,0) node[above]{\footnotesize $X_1.op_5$};

      \draw (3.0,1.0) node{$\bullet$} +(-0.5,0) node[above]{\footnotesize $X_2.op_1$};
      \draw (4.7,1.0) node{$\bullet$} +(-0.1,0) node[above]{\footnotesize $X_2.op_2$};
      \draw (3.4,0.5) node{$\bullet$} +(0.2,0) node[below]{\footnotesize $X_2.op_3$};
      \draw (3.2,0.0) node{$\bullet$} node[below]{\footnotesize $X_2.op_4$};
      \draw (4.7,0.0) node{$\bullet$} node[below]{\footnotesize $X_2.op_5$};
      \draw (6.2,0.0) node{$\bullet$} node[below]{\footnotesize $X_2.op_6$};
      \draw (7.7,0.0) node{$\bullet$} node[below]{\footnotesize $X_2.op_7$};

      \draw (5.7,1.0) node{$\bullet$} +(0.6,-0.1) node[above]{\footnotesize $X_3.op_1$};
      \draw (4.6,0.5) node{$\bullet$} +(0.3,0.05) node[below]{\footnotesize $X_3.op_2$};
      \draw (7.0,1.0) node{$\bullet$} +(0.4,0) node[below]{\footnotesize $X_3.op_3$};
      
    \end{tikzpicture}
  \caption{As the ordering between different objects follows the process order, that is contained into the 
  serialization order of each object, no loop can appear.}
  \label{fig:rounds:hist}
  \end{subfigure}
  \hspace{\fill}
  \caption{The composition of sequentially consistent objects used in different rounds is sequentially consistent.}
  \label{fig:rounds}
\end{figure}

\begin{proposition}\label{prop:round}
  Let $(T_r)_{r\in\mathbb{N}}$ be a family of sequential specifications and $(X_r)_{r\in\mathbb{N}}$ be a 
  family of shared objects such that, for all $r$, $X_r$ is $SC(T_r)$-consistent. Let $H$ be a history
  that does not contain two operations $X_r.a$ and $X_{r'}.b$ with $r>r'$ such that $X_r.a$ precedes $X_{r'}.b$ in the process order. 
  Then $H$ is sequentially consistent with respect to the composition of all the $T_r$.
\end{proposition}

\begin{proof}
  Let $(T_r)_{r\in\mathbb{N}}$ be a family of sequential specifications and $(X_r)_{r\in\mathbb{N}}$ be a 
  family of shared object such that, for all $r$, $X_r$ is $SC(T_r)$-consistent. Let $H$ be a history
  that does not contain two operations $X_r.a$ and $X_{r'}.b$ with $r>r'$ such that $X_r.a$ precedes $X_{r'}.b$ in the process order. 
  
  For each $X_r$, there exists a linearization $l_r$ that contains the operations on $X_r$ and respects $T_r$.
  For each operation $op$, let us denote by $op.r$ the index of the object $X_r$ on which it is made
  and by $op.i$ the number of operations that precede $op$ in the linearization $l_r$.
  Let us define two binary relations $\lessdot$ and $\rightarrow$ on the operations of $H$. For two operations $op$ and $op'$,
  $op \lessdot op'$ if $op.r < op'.r$, or $op.r = op'.r$ and $op.i \le op'.i$. Note that $\lessdot$ is the concatenation of all 
  the linear extensions, so it is a total order on all the operations of $H$, but it may not be a linear extension as an
  operation can have an infinite past if a process does not finish its round. For two operations $op$ and $op'$, $op \rightarrow op'$ if
  $op$ and $op'$ were done in that order by the same process, or $op.r = op'.r$ and $op.i \le op'.i$. Let $\rightarrow^\star$
  be the transitive closure of $\rightarrow$.

  Notice that, according to the round based model, $\rightarrow$ is contained into $\lessdot$, and so is $\rightarrow^\star$ because $\lessdot$ is transitive.
  The relation $\rightarrow^\star$ is transitive and reflexive by construction. Moreover, if $op\rightarrow^\star op' \rightarrow^\star op$, we have $op.r\le op'.r \le op.r$
  and therefore $op.i\le op'.i \le op.i$, so $op=op'$ (antisymmetry), which proves that $\rightarrow^\star$ is a partial order. Moreover, let us suppose that an
  operation contains an infinite past according to $\rightarrow^\star$. There is a smallest such operation, $op_{\min}$, according to $\lessdot$. The direct predecessors of
  $op_{\min}$ according to $\rightarrow$ are smaller than $op_{\min}$ according to $\lessdot$, so they have a finite past. Moreover, they precede $op_{\min}$
  either in the process order or in the linearization $l_{op.r}$, so there is a finite number of them. This is a contradiction, so all operations have a finite
  past according to $\rightarrow^\star$.
  It is possible to extend $\rightarrow^\star$ to a total order $\leq$ such that all operations have a finite past according to $\leq$. As $\leq$ contains
  the total orders defined by all the $l_r$, the execution of all the operations in the order $\leq$ respects the sequential specification of the composition
  of all the $X_r$.
\end{proof}

%%%%%%%%%%%%%%%%%%%%
\section{Implementation of a Sequentially Consistent Memory}\label{sec:implementation}
%%%%%%%%%%%%%%%%%%%%

In this section we will describe the computation model that we consider for the implementation of a sequentially consistent shared memory (Section~\ref{sec:model}). 
In Section~\ref{sec:memory} we will discuss the characteristics of such a memory and, finally, in Section~\ref{sec:algo} we will present the proposed implementation of the discussed data structure. 
Finally, in Section~\ref{sec:complexity} we discuss the complexity of the proposed implementation. 

%%%%%%%%%%%%%%%%%%%%
\subsection{Computation Model}
\label{sec:model}
%%%%%%%%%%%%%%%%%%%%

The computation system consists of a set $\Pi$ of $n$ sequential processes which are denoted $p_0, p_1, \ldots, p_{n-1}$. 
The processes are asynchronous, in the sense that they all proceed at their own speed, 
not upper bounded and unknown to all other processes. 

Among these $n$ processes, up to $t$ may crash (halt prematurely) but otherwise execute correctly the algorithm until the moment of their crash. 
We call a process \emph{faulty} if it crashes, otherwise it is called \emph{correct} or \emph{non-faulty}. 
In the rest of the paper we will consider the above model restricted to the case $t<\frac{n}{2}$.

The processes communicate with each other by sending and receiving messages through a complete network of bidirectional communication channels. 
This means that a process can directly communicate with any other process, including itself ($p_i$ receives its own messages instantaneously), 
and can identify the sender of the message it received. Each process is equipped with two operations: \textbf{send} and \textbf{receive}. 

The channels are reliable (no losses, no creation, no duplication, no alteration of messages) and asynchronous 
(finite time needed for a message to be transmitted but there is no upper bound). We also assume 
the channels are FIFO: if $p_i$ sends two messages to $p_j$, $p_j$ will receive them in the order they were sent.
As stated in \cite{birman1987reliable}, FIFO channels can always be implemented on top of non-FIFO channels. 
Therefore, this assumption does not bring additional computational power to the model, but it allows us to simplify the writing of the algorithm. Process $p_i$ can also use the macro-operation \textbf{FIFO broadcast}, that can be seen as a multi-send that sends a message to all processes, including itself. Hence, if a faulty process crashes during the broadcast operation some processes may receive the message while others may not, otherwise all correct processes will eventually receive the message.

%%%%%%%%%%%%%%%%%%%%
\subsection{Single-Writer/Multi-Reader Registers and Snapshot Memory}
\label{sec:memory}
%%%%%%%%%%%%%%%%%%%%

The shared memory considered in this paper, called a \emph{snapshot memory}, consists of an array of shared registers denoted $\REG[1..n]$. Each entry $\REG[i]$ represents a single-writer/multi-reader (SWMR) register. When process $p_i$ invokes $\REG.\AlgoWrite(v)$, the value $\Algov$ is written into the SWMR register $\REG[i]$ associated with process $p_i$. Differently, any process $p_i$ can read the whole array $\REG$ by invoking a single operation namely $\REG.\AlgoSnap()$. According to the sequential specification of the snapshot memory, $\REG.\AlgoSnap()$ returns an array containing the most recent value written by each process or the initial default value if no value is written on some register. 
Concurrency is possible between snapshot and writing operations, as soon as the considered consistency criterion, namely linearizability or sequential consistency, is respected. Informally in a sequentially consistent snapshot memory, each snapshot operation must return the last value written by the process that initiated it,
and for any pair of snapshot operations, one must return values at least as recent as the other for all registers.

Compared to read and write operations, the snapshot operation is a higher level abstraction introduced in \cite{Afek93} that eases program design without bringing additional power with respect to shared registers. Of course this induces an additional cost: the best known simulation, above SWMR registers proposed in \cite{AttiyaR98}, needs $O(n\log n)$ basic read/write operations to implement each of the snapshot and the associated update operations.

Since the seminal paper \cite{attiya1995sharing} that proposed the so-called ABD simulation that emulates a linearizable shared memory over a message-passing distributed system, most of the effort has been put on the shared memory model given that a simple stacking allows to translate any shared memory-based result to the message-passing system model. 
Several implementations of linearizable snapshot have been proposed in the literature some works consider variants of snapshot (e.g. immediate snapshot \cite{BorowskyG92}, weak-snapshot \cite{dwork1992time}, one scanner \cite{KirousisST94}) others consider that special constructions such as test-and-set (T\&S) \cite{AHR95} or load-link/store-conditional (LL/SC) \cite{RST01} are available, the goal being to enhance time and space efficiency.
In this paper, we propose the first message-passing sequentially consistent (not linearizable) snapshot memory implementation directly over a message-passing system (and consequently the first sequentially consistent array of SWMR registers), as traditional read and write operations can be immediately deduced from snapshot and update with no additional cost.

%%%%%%%%%%%%%%%%%%%%
\subsection{The Proposed Algorithm}
\label{sec:algo}
%%%%%%%%%%%%%%%%%%%%

\begin{algorithm}[p]
  \tcc{Local variable initialization}
  $\KwX_i  \leftarrow [0,\dots, 0]$\label{al:SCS:varX}\tcp*{$\KwX_i \in \mathbb{N}^{n}$: $\KwX_i[j]$ is the last validated value written by $p_j$}
  $\KwVC_i \leftarrow [0,\dots, 0]$\label{al:SCS:varVC}\tcp*{$\KwVC_i \in \mathbb{N}^{n}$: $\KwVC_i[j]$ is the stamp given by $p_j$ to value $\KwX_i[j]$}
  $\KwSC_i \leftarrow 0$\label{al:SCS:varSC}\tcp*{$\KwSC_i \in \mathbb{N}$: used to stamp all the updates}
  $\KwG_i  \leftarrow \emptyset$\label{al:SCS:varG}\tcp*{$\KwG_i \subset \mathbb{N}^{3+n}$: contains a $g = (g.\KwGV, g.\KwGK, g.\KwGT, g.\KwGCL)$ per non-val. update of $g.\KwGV$ by $p_{g.\KwGK}$}
  $\KwV_i  \leftarrow \bot$\label{al:SCS:varV}\tcp*{$\KwV_i \in \mathbb{N}\cup \{\bot\}$: stores updates that have not yet been proposed to validation}

  \nonl\hrulefill\\
  \nonl\SubAlgo{\Op $\KwWrite(\Kwv)$  \tcc*[h]{$\Kwv\in \mathbb{N}$: written value; no return value}}{
    \uIf(\label{al:SCS:w1}\tcp*[f]{no non-validated update by $p_i$}){$\forall g\in \KwG_i: g.\KwGK\neq i$} {
%      $\KwSC_i\leftarrow \KwSC_i + 1$;\label{al:SCS:w2}
      $\KwSC_i\text{++}$\; \Broadcast $\KwM(\Kwv,i,\KwSC_i,\KwSC_i)$\;\label{al:SCS:w3}
    }
    \lElse(\label{al:SCS:w4}\tcp*[f]{postpone the update}){$\KwV_i \leftarrow \Kwv$}
  }

  \nonl\hrulefill\\
  \nonl\SubAlgo{\Op $\KwSnap()$  \tcc*[h]{return type: $\mathbb{N}^n$}}{
    \Wait{$\KwV_i=\bot \land \forall g\in \KwG_i: g.\KwGK\neq i$} \label{al:SCS:r1}\tcp*[r]{make sure $p_i$'s updates are validated}
    \Return $\KwX_i$\label{al:SCS:r2}\;
  }

  \nonl\hrulefill\\
  \nonl\SubAlgo{\RecA $\KwM(\Kwv,\Kwk,\Kwt,\Kwcl)$ \RecB $p_j$} {
    \tcc{$\Kwv\in \mathbb{N}$: written value, $\Kwk\in \mathbb{N}$: writer id, $\Kwt\in \mathbb{N}$: stamp by $p_{\Kwk}$, $\Kwcl\in \mathbb{N}$: stamp by $p_j$}
    \If(\label{al:SCS:mA1}\tcp*[f]{update not validated yet}){$\Kwt>\KwVC_i[\Kwk]$} {
      \eIf(\label{al:SCS:mA2}\tcp*[f]{update already known}){$\exists g\in \KwG_i: g.\KwGK = \Kwk \land g.\KwGT = \Kwt$}{
        $g.\KwGCL[j] \leftarrow \Kwcl$\;\label{al:SCS:mA3}
      }(\label{al:SCS:mA4}\tcp*[f]{first message for this update}){
        \If{$\Kwk\neq i$}{
            $\KwSC_i\text{++}$\; \Broadcast $\KwM(\Kwv, \Kwk, \Kwt, \KwSC_i)$\label{al:SCS:mA6}\tcp*[r]{forward with own stamp}
        }
        \textbf{var} $g \leftarrow \left(g.\KwGV = \Kwv, g.\KwGK=\Kwk, g.\KwGT=\Kwt, g.\KwGCL=[\infty,\dots,\infty]\right)$\label{al:SCS:mA7}\;
        $g.\KwGCL[j] \leftarrow \Kwcl$\label{al:SCS:mA8}\;
        $\KwG_i \leftarrow \KwG_i\cup \{g\}$\label{al:SCS:mA9}\tcp*[r]{create an entry in $\KwG_i$ for the update}
      }
    }
    \textbf{var} $G' = \{g \in \KwG_i: |\{l : g'.\KwGCL[l] < \infty\}| > \frac{n}{2} \} $\label{al:SCS:mB1}\tcp*[r]{$G'$ contains updates that can be validated}
    \lWhile(\label{al:SCS:mB2}){$\exists g\in \KwG_i\setminus G', g'\in G' : |\{l : g'.\KwGCL[l] < g.\KwGCL[l] \}| \neq \frac{n}{2}$}{
      $G'\leftarrow G'\setminus \{g'\}$\label{al:SCS:mB3}}
    $\KwG_i \leftarrow \KwG_i \setminus G'$\label{al:SCS:mB4}\tcp*[r]{validate updates of $G'$}
    \For(\label{al:SCS:mB5}){$g\in G'$}{
      \lIf{$\KwVC_i[g.\KwGK]<g.\KwGT$}{$\KwVC_i[g.\KwGK]=g.\KwGT;$ $\KwX_i[g.\KwGK]=g.\KwV$}\label{al:SCS:mB6}
    }
    \If(\label{al:SCS:mC1}\tcp*[f]{start validation process for postponed update if any}){$\KwV_i\neq\bot \land \forall g\in \KwG_i: g.\KwGK\neq i$}{
%      $\KwSC_i\leftarrow \KwSC_i + 1$\label{al:SCS:mC2}\;
      $\KwSC_i\text{++}$\;
      \Broadcast $\KwM(\KwV_i,i,\KwSC_i,\KwSC_i)$\label{al:SCS:mC3}\;
      $\KwV_i \leftarrow \bot$\label{al:SCS:mC4}\;
    }
  }
  \caption{Implementation of a sequentially consistent memory (code for $p_i$)}
  \label{algo:SCS}
\end{algorithm}

%DESCRIPTION OF THE ALGORITHM

Algorithm~\ref{algo:SCS} proposes an implementation of the sequentially consistent snapshot memory data structure presented in Section~\ref{sec:memory}. Process $p_i$ can write a value $\Algov$ in its own register $\REG[i]$
by calling the operation $REG.\AlgoWrite(v)$, implemented by the lines~\ref{al:SCS:w1}-\ref{al:SCS:w4}. It can also call the operation $REG.\AlgoSnap()$, implemented by the lines~\ref{al:SCS:r1}-\ref{al:SCS:r2}.
Roughly speaking, the principle of this algorithm is to maintain, on each process, a local view of the object that reflects a set of \emph{validated} update operations. 
To do so, when a value is written, all processes label it with their own timestamp. The order in which processes timestamp two different update operations define a \emph{dependency relation} between these operations. For two operations $a$ and $b$, if $b$ depends on $a$, then $p_i$ cannot validate $b$ before $a$. 

More precisely, each process $p_i$ maintains five local variables:
\begin{itemize}
\item $\AlgoX_i \in \mathbb{N}^n$ represents the array of most recent validated values written on each register.
\item $\AlgoVC_i \in \mathbb{N}^n$ represents the timestamps associated with the values stored in $\AlgoX_i$, labelled by the process that initiated them.
\item $\AlgoSC_i\in \mathbb{N}$ is an integer clock used by $p_i$ to timestamp all the update operations. $\AlgoSC_i$ is incremented each time a message is sent,
    %using the C-like notation $\AlgoSC_i\text{++}$
    which ensures all timestamps from the same process are different. 
\item $\AlgoG_i \subset \mathbb{N}^{3+n}$ encodes the dependencies between the update
    operations that have not been validated yet, as they are known by $p_i$. An element $\Algog\in \AlgoG_i$, of the form $(\Algog.\AlgoGV, \Algog.\AlgoGK, \Algog.\AlgoGT, \Algog.\AlgoGCL)$, represents
    the update operation of value $\Algog.\AlgoGV$ by process $p_{\Algog.\AlgoGK}$ labelled by process $p_{\Algog.\AlgoGK}$ with timestamp $\Algog.\AlgoGT$. For all $0\leq j<n$,
    $\Algog.\AlgoGCL[j]$ contains the timestamp associated by $p_j$ if it is known by $p_i$, and $\infty$ otherwise. 
    
    All updates of a history can be uniquely represented by a pair of integers $(k, t)$, where $p_k$ is the process that invoked it, and $t$ is the timestamp associated 
    to this update by $p_k$. Considering a history and a process $p_i$, we define the dependency relation $\rightarrow_i$ on pairs of integers $(k, t)$, by 
    $(k, t)\rightarrow_i (k', t')$ if for all $\Algog, \Algog'$ ever inserted in $G_i$ with
    $(\Algog.\AlgoGK,\Algog.\AlgoGT)=(k,t)$, $(\Algog'.\AlgoGK,\Algog'.\AlgoGT)=(k',t')$, 
    we have $|\{j : \Algog'.\AlgoGCL[j] < \Algog.\AlgoGCL[j] \}| \le \frac{n}{2}$ (i.e. the dependency does not exist if $p_i$ knows that a majority of processes have seen the first update before the second).
    Let $\rightarrow_i^\star$ denote the transitive closure of $\rightarrow_i$.
\item $\AlgoV_i\in \mathbb{N}\cup\{\bot\}$ is a buffer register used to store a value written while the previous one is not yet validated. This is necessary for validation (see below).
\end{itemize}

The key of the algorithm is to ensure the inclusion between sets of validated updates on any two processes at any time. Remark that it is not always necessary 
to order all pairs of update operations to implement a sequentially consistent snapshot memory: for example, two update operations on different registers commute. 
Therefore, instead of validating both operations on all processes in the exact same order (which requires Consensus), we can validate them at the same 
time to prevent a snapshot to occur between them. Therefore, it is sufficient to ensure that, for all pairs of update operations, there is a dependency 
agreed by all processes (possibly in both directions). This property is expressed by Lemma~\ref{lemma:safety} from Section~\ref{sec:proof}.

This is done by the mean of messages of the form $\AlgoM(\Algov,\Algok,\Algot,\Algocl)$ containing four integers: $\Algov$ the value written, $\Algok$ the identifier of the process that initiated the update, $\Algot$ the timestamp given by $p_{\Algok}$ and $\Algocl$ the timestamp given by the process that sent this message. Timestamps 
of successive messages sent by $p_i$ are unique and totally ordered, thanks to variable $\AlgoSC_i$, that is incremented each time a message is sent by $p_i$.
When process $p_i$ wants to submit a value $\Algov$ for validation, it FIFO-broadcasts a message $\AlgoM(\Algov,i,\AlgoSC_i,\AlgoSC_i)$ (lines~\ref{al:SCS:w3} and~\ref{al:SCS:mC3}). 
When $p_i$ receives a message $\AlgoM(\Algov,\Algok,\Algot,\Algocl)$, three cases are possible. If $p_i$ has already validated the corresponding update ($\Algot > \AlgoVC_i[\Algok]$),
the message is simply ignored. Otherwise, if it is the first time $p_i$ receives a message concerning this update ($\AlgoG_i$ does not contain any piece of information concerning it), 
it FIFO-broadcasts a message with its own timestamp and adds a new entry $\Algog\in \AlgoG_i$. Whether it is its first message or not, $p_i$ records the timestamp $\Algocl$, given by $p_j$, in $\Algog.\AlgoGCL[j]$ (lines~\ref{al:SCS:mA3} or~\ref{al:SCS:mA8}). Note that we cannot update $\Algog.\AlgoGCL[\Algok]$ at this point, as the broadcast is not causal: if $p_i$ did so, it could miss dependencies imposed by the order in which $p_{\Algok}$ saw concurrent updates. Then, $p_i$ tries to validate update operations: $p_i$ can validate an operation $a$ if it has received messages from a majority of processes, and there is no operation $b\rightarrow_i^\star a$ that cannot be validated. For that, it creates the set $G'$ that initially contains all the operations that have received enough messages, and removes all operations with unvalidatable dependencies from it (lines~\ref{al:SCS:mB1}-\ref{al:SCS:mB3}), and then updates $\AlgoX_i$ and $\AlgoVC_i$ with the most recent validated values (lines~\ref{al:SCS:mB4}-\ref{al:SCS:mB6}).

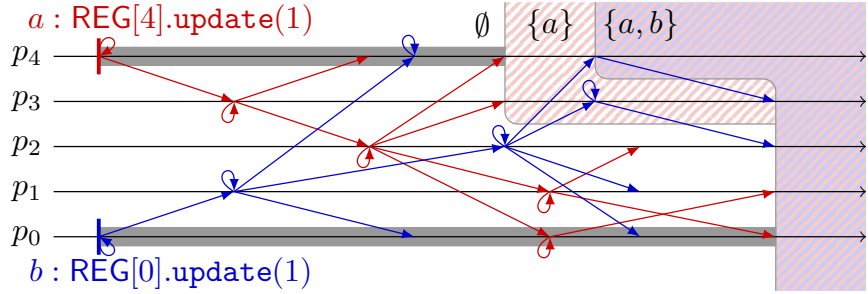
\begin{figure}[t]
%\begin{subfigure}[b]{0.35\textwidth}
    \centering
\scalebox{1.2}{
\begin{tikzpicture}

%\draw[step=1cm,gray,very thin] (-2,-1) grid (8,3);

\fill[colorB!20] (5,2.6) {[rounded corners=5] -- (5,1.75) -- (7,1.75)} -- (7,-0.6) -- (8,-0.6) -- (8,2.6) -- cycle;
\fill[pattern=stripes, pattern color=colorA!20] (4,2.6) {[rounded corners=5] -- (4,1.25)} -- (7,1.25) -- (7,-0.6) -- (8,-0.6) -- (8,2.6) -- cycle;
\draw[black!40,rounded corners] (5,2.6) -- (5,1.75) -- (7,1.75) -- (7,-0.6);
\draw[black!40,rounded corners] (4,2.6) -- (4,1.25) -- (7,1.25);
\draw[draw=black!40,fill=black!40] (-0.5,-0.1) rectangle (7,0.1);
\draw[draw=black!40,fill=black!40] (-0.5,1.9) rectangle (4,2.1);

\draw[->] (-1,2.0) node[left]{$p_4$} -- (8,2.0);
\draw[->] (-1,1.5) node[left]{$p_3$} -- (8,1.5);
\draw[->] (-1,1.0) node[left]{$p_2$} -- (8,1.0);
\draw[->] (-1,0.5) node[left]{$p_1$} -- (8,0.5);
\draw[->] (-1,0.0) node[left]{$p_0$} -- (8,0.0);

\draw (3.75,2.35) node{$\emptyset$};
\draw (4.5,2.35) node{$\{a\}$};
\draw (5.5,2.35) node{$\{a, b\}$};

\draw[colorA,  very thick] (-0.5,2.0) +(0,-0.2) -- +(0,0.2) +(0.8,0.4) node{$a: \REG[4].\AlgoWrite(1)$};
\draw[colorB, very thick] (-0.5,0.0) +(0,-0.2) -- +(0,0.2) +(0.8,-0.4) node{$b: \REG[0].\AlgoWrite(1)$};

\draw[-latex, colorA] (-0.5,2.0) to[out=90,in=30,distance=10] (-0.5,2.0);
\draw[-latex, colorA] (-0.5,2.0) -- (1,1.5);
\draw[-latex,colorA] (1,1.5) to[out=-150,in=-90,distance=10] (1,1.5);
\draw[-latex,colorA] (1,1.5) -- (2.5,1.0);
\draw[-latex,colorA] (1,1.5) -- (2.5,2.0);
\draw[-latex,colorA] (2.5,1.0) to[out=-150,in=-90,distance=10] (2.5,1.0);
\draw[-latex,colorA] (2.5,1.0) -- (4.0,1.5);
\draw[-latex,colorA] (2.5,1.0) -- (4.0,2.0);
\draw[-latex,colorA] (2.5,1.0) -- (4.5,0.5);
\draw[-latex,colorA] (2.5,1.0) -- (4.5,0.0);
\draw[-latex,colorA] (4.5,0.5) to[out=-150,in=-90,distance=10] (4.5,0.5);
\draw[-latex,colorA] (4.5,0.5) -- (7.0,0.0);
\draw[-latex,colorA] (4.5,0.5) -- (5.50,1.0);
\draw[-latex,colorA] (4.5,0.0) to[out=-150,in=-90,distance=10] (4.5,0.0);
\draw[-latex,colorA] (4.5,0.0) -- (7.0,0.5);

\draw[-latex,colorB] (-0.5,0.0) to[out=-90,in=-30,distance=10] (-0.5,0.0);
\draw[-latex,colorB] (-0.5,0.0) -- (1,0.5);
\draw[-latex,colorB] (1.0,0.5) to[out=150,in=90,distance=10] (1.0,0.5);
\draw[-latex,colorB] (1.0,0.5) -- (4.0,1.0);
\draw[-latex,colorB] (1.0,0.5) -- (3.0,0.0);
\draw[-latex,colorB] (1.0,0.5) -- (3.0,2.0);
\draw[-latex,colorB] (4.0,1.0) to[out=150,in=90,distance=10] (4.0,1.0);
\draw[-latex,colorB] (4.0,1.0) -- (5.5,0.0);
\draw[-latex,colorB] (4.0,1.0) -- (5,1.5);
\draw[-latex,colorB] (4.0,1.0) -- (5,2.0);
\draw[-latex,colorB] (4.0,1.0) -- (5.5,0.5);
\draw[-latex,colorB] (5.0,1.5) to[out=150,in=90,distance=10] (5.0,1.5);
\draw[-latex,colorB] (5.0,1.5) -- (7.0,1.0);
\draw[-latex,colorB] (3,2.0) to[out=150,in=90,distance=10] (3,2.0);
\draw[-latex,colorB] (5.0,2.0) -- (7.0,1.5);

\end{tikzpicture}

}
\caption{An execution of Algorithm~\ref{algo:SCS}. An update is validated by a process when it has received enough messages for this update, and all the other updates it depends of have also been validated.}
\label{fig:expl_algo:handshake}
%\end{subfigure}
%\hspace{\fill}
\end{figure}

This mechanism is illustrated in Fig. \ref{fig:expl_algo:handshake}, featuring five processes. Processes $p_0$ and $p_4$ initially call operation $\REG.\AlgoWrite(1)$. Messages that have an impact in the algorithm are represented by arrows, and messages that do not appear on the figure are received later. 
Several situations may occur. The simplest case is process $p_3$, that received three messages concerning $a$ (from $p_4$, $p_3$ and $p_2$, with $3>\frac{n}{2}$) before its first message concerning $b$, allowing it to validate $a$. The case of process $p_4$ is similar: even if it knows that process $p_1$ saw $b$ before $a$, it received messages concerning $a$ from three \emph{other} processes, which allows it to ignore the message from $p_1$. 
At first sight, the situation of processes $p_0$ and $p_1$ may look similar to the situation of $p_4$. However, the message they received concerning $a$ 
and one of the messages they received concerning $b$ come from the same process $p_2$, which forces them to respect the dependency $a\rightarrow_0 b$.
Note that the same situation occurs on process $p_2$ so, even if $a$ has been validated before $b$ by other processes, $p_2$ must respect the dependency 
$b\rightarrow_2 a$. 

Sequential consistency requires the total order to contain the process order. Therefore, a snapshot of process $p_i$ must return values at least as recent as its last updated value.
In other words, it is not allowed to return from a snapshot between an update and the time when it is validated 
(grey zones in Fig.~\ref{fig:expl_algo:handshake}). There are two ways to implement this: we can either wait at the end of each update until it is validated, in 
which case all snapshot operations are done for free, or wait at the beginning of all snapshot operations that immediately follow an update operation. This
extends the remark of~\cite{attiya1994sequential} to crash-prone asynchronous systems: to implement a sequentially consistent memory, it is necessary and
sufficient to wait either during read or during write operations. In Algorithm~\ref{algo:SCS}, we chose to wait during read/snapshot operations (line~\ref{al:SCS:r1}). 
This is more efficient for two reasons: first, it is not necessary to wait between two consecutive updates, which can not be avoided if we 
wait at the end of the update operation, and second the time between the end of an update and the beginning of a snapshot counts in the validation 
process, but it can be used for local computations. Note that when two snapshot operations are invoked successively, the second one also returns immediately, 
which improves the result of~\cite{attiya1994sequential} according to which waiting is necessary for all the operations of one kind.

\begin{figure}[t]
%\begin{subfigure}[b]{0.63\textwidth}
    \centering
\scalebox{1.1}{
\begin{tikzpicture}

\colorlet{colorA}{black!20!red}
\colorlet{colorB}{black!20!blue}

\draw[colorA] (1.5,-0.85) node{$a$};
\draw         (2  ,-0.85) node{$\rightleftharpoons$};
\draw[colorB] (2.5,-0.85) node{$b$};
\draw         (3  ,-0.85) node{$\rightleftharpoons$};
\draw[colorA] (3.5,-0.85) node{$c$};
\draw         (4  ,-0.85) node{$\rightleftharpoons$};
\draw[colorB] (4.5,-0.85) node{$d$};
\draw         (5  ,-0.85) node{$\rightleftharpoons$};
\draw[colorA] (5.5,-0.85) node{$e$};
\draw         (6  ,-0.85) node{$\rightleftharpoons$};
\draw[colorB] (6.5,-0.85) node{$f$};
\draw         (7  ,-0.85) node{$\rightleftharpoons$};
\draw[colorA] (7.5,-0.85) node{$g$};
\draw         (8  ,-0.85) node{$\rightleftharpoons$};
\draw[colorB] (8.5,-0.85) node{$h$};
\draw         (9  ,-0.85) node{$\rightleftharpoons$};
\draw         (9.5,-0.85) node{$\dots$};

\draw[->] (0.5,1.5) node[left]{$p_3$} -- (9.5,1.5);
\draw[->] (0.5,1.0) node[left]{$p_2$} -- (9.5,1.0);
\draw[->] (0.5,0.5) node[left]{$p_1$} -- (9.5,0.5);
\draw[->] (0.5,0.0) node[left]{$p_0$} -- (9.5,0.0);

\draw[colorA, very thick] (1,1.5) +(0,-0.2) -- +(0,0.2) +(0,0.35) node{$a$};
\draw[colorA, very thick] (3,1.5) +(0,-0.2) -- +(0,0.2) +(0,0.35) node{$c$};
\draw[colorA, very thick] (5,1.5) +(0,-0.2) -- +(0,0.2) +(0,0.35) node{$e$};
\draw[colorA, very thick] (7,1.5) +(0,-0.2) -- +(0,0.2) +(0,0.35) node{$g$};
\draw[colorA] (9,1.85) node{$\dots$};

\draw[colorB, very thick] (1,0.0) +(0,-0.2) -- +(0,0.2) +(0,-0.35) node{$b$};
\draw[colorB, very thick] (3,0.0) +(0,-0.2) -- +(0,0.2) +(0,-0.35) node{$d$};
\draw[colorB, very thick] (5,0.0) +(0,-0.2) -- +(0,0.2) +(0,-0.35) node{$f$};
\draw[colorB, very thick] (7,0.0) +(0,-0.2) -- +(0,0.2) +(0,-0.35) node{$h$};
\draw[colorB] (9,-0.35) node{$\dots$};

\draw[-latex,colorA] (1,1.5) -- (1.67,1.0);
\draw[-latex,colorA] (3,1.5) -- (3.67,1.0);
\draw[-latex,colorB] (1,0.0) -- (4.33,1.0);
\draw[-latex,colorA] (5,1.5) -- (5.67,1.0);
\draw[-latex,colorB] (3,0.0) -- (6.33,1.0);
\draw[-latex,colorA] (7,1.5) -- (7.67,1.0);
\draw[-latex,colorB] (5,0.0) -- (8.33,1.0);
\draw[-latex,colorB] (7,0.0) -- (9.0,1.0);

\draw[-latex,colorB] (1,0.0) -- (1.67,0.5);
\draw[-latex,colorB] (3,0.0) -- (3.67,0.5);
\draw[-latex,colorA] (1,1.5) -- (4.33,0.5);
\draw[-latex,colorB] (5,0.0) -- (5.67,0.5);
\draw[-latex,colorA] (3,1.5) -- (6.33,0.5);
\draw[-latex,colorB] (7,0.0) -- (7.67,0.5);
\draw[-latex,colorA] (5,1.5) -- (8.33,0.5);
\draw[-latex,colorA] (7,1.5) -- (9.0,0.5);
\end{tikzpicture}
}
\caption{If we are not careful, infinite chains of dependencies may occur. We must avoid infinite chains of dependencies in order to ensure termination}
\label{fig:expl_algo:dependences}
%\end{subfigure}
%\caption{An update is validated by a process when it has received enough messages for this update, and all the other updates it depends of have also been 
%validated (Fig.~\ref{fig:expl_algo:handshake}). We must avoid infinite chains of dependencies in order to ensure termination (Fig.~\ref{fig:expl_algo:dependences}).\vspace{-4mm}}
%\label{fig:expl_algo}
\end{figure}
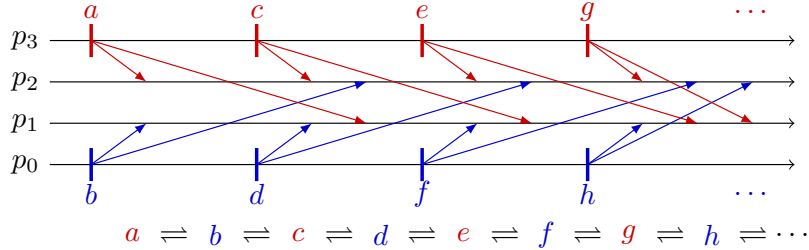

In order to obtain termination of the snapshot operations (and progress in general), it is necessary to ensure that all update operations are eventually
validated by all processes. This property is expressed by Lemma~\ref{lemma:liveness} from Section~\ref{sec:proof}.
Figure \ref{fig:expl_algo:dependences} illustrates what could happen. On the one hand, 
process $p_2$ receives a message concerning $a$ and a message concerning $c$ before a message concerning $b$. On the other hand, 
process $p_1$ receives a message concerning $b$ before messages concerning $a$ and $c$. Therefore, it may create dependencies 
$a\rightarrow_i b \rightarrow_i c \rightarrow_i b \rightarrow_i a$ on some process $p_i$, which means $p_i$ will be forced to validate $a$ and $c$
at the same time, even if they are ordered by the process order. The pattern in Fig.~\ref{fig:expl_algo:dependences} shows that it can 
result in an infinite chain of dependencies, blocking validation of any update operation. To break this chain, we force process $p_3$ to wait until $a$ is validated
locally before it proposes $c$ to validation, by storing the value written by $c$ in a local variable $\AlgoV_i$ until $a$ is validated (lines~\ref{al:SCS:w1} and~\ref{al:SCS:w4}). When $a$ is validated, 
we start the same validation process for $c$ (lines~\ref{al:SCS:mC1}-\ref{al:SCS:mC4}).
Remark that, if several updates (say $c$ and $e$) happen before $a$ is validated, the update of $c$ can be dropped as it will eventually be overwritten by $e$. 
In this case, $c$ will happen just before $e$ in the final linearization required for sequential consistency. 

This algorithm could be adjusted to implement multi-writer/multi-reader registers. Only three points must be changed. First, the identifier of the register written 
should be added to all messages and all $\Algog\in \AlgoG_i$. Second, concurrent updates on the same register must be ordered; this can be done, 
for example, by replacing $\AlgoSC_i$ by a Lamport Clock, that respects the order in which updates are validated, and using a lexicographic 
order on pairs $(\Algocl, \Algok)$. Third, variable $\AlgoV_i$ must be replaced by a set of update operations, and so does the value contained in the messages. 
All in all, this greatly complexifies the algorithm, without changing the way concurrency is handled. 
This is why we only focus on collections of SWMR registers here.

\subsection{Correctness}\label{sec:proof}

In order to prove that Algorithm~\ref{algo:SCS} implements a sequentially consistent snapshot memory, we must show that two important properties are verified 
by all histories it admits. These two properties correspond to lemmas~\ref{lemma:safety} and~\ref{lemma:liveness}. In Lemma~\ref{lemma:safety},
we show that it is possible to totally order the sets of updates validated by two processes at different moments. This allows us to build a total order
on all the operations. In Lemma~\ref{lemma:liveness}, we prove that all update operations are eventually validated by all processes. This is important 
to ensure termination of snapshot operations, and to ensure that update operations can not be ignored forever. 
Before that, Lemma~\ref{lemma:broadcast} expresses a central property on how the algorithm works: the fact that each correct process 
broadcasts a message corresponding to each written value proposed to validation.
Finally, Property~\ref{prop:correct} proves that all histories admitted by Algorithm~\ref{algo:SCS} are sequentially consistent.

In the following and for each process $p_i$ and local variable $x_i$ used in the algorithm, let us denote by $x_i^t$ the value of $x_i$ at time $t$.
For example, $\AlgoVC_i^0$ is the initial value of $\AlgoVC_i$. For arrays of $n$ integers $cl$ and $cl'$, we also denote by $cl \le cl'$ 
the fact that, for all $i$, $cl[i]\le cl'[i]$ and $cl < cl'$ if $cl \le cl'$ and $cl \neq cl'$.

\begin{lemma} \label{lemma:broadcast}
If a message $\AlgoM(\Algov, \Algok, \Algot, \Algocl)$ is broadcast by a correct process $p_i$, 
then each correct process $p_j$ broadcasts a unique message $\AlgoM(\Algov, \Algok, \Algot, \Algocl')$.

In the following, for all processes $p_j$ and pairs $(k, t)$, let us denote by $M_j(k,t)$ the message 
$\AlgoM(\Algov, \Algok, \Algot, \Algocl')$ and by $CL_j(k,t) = \Algocl'$ the stamp that $p_j$ put in this message. 
\end{lemma}
\begin{proof}
Let $p_i$ and $p_j$ be two correct processes, and suppose $p_i$ broadcasts a message $M_i(\Algok, \Algot)$.

First, we prove that $p_j$ broadcasts a message $M_j(\Algok, \Algot)$.
As $p_i$ is correct, $p_j$ will eventually receive the message sent by $p_i$. At that time, 
if $\Algot > \AlgoVC_j[\Algok]$, after the condition on line~\ref{al:SCS:mA2} and whatever its result,
$\AlgoG_i$ contains a value $\Algog$ with $\Algog.\AlgoGK = \Algok$ and $\Algog.\AlgoGT = \Algot$. That $g$ was inserted on 
line~\ref{al:SCS:mA2} (possibly after the reception of a different message), just after $p_j$ sent a message $M_j(\Algok, \Algot)$ 
at line~\ref{al:SCS:mA6}.
Otherwise, $\AlgoVC_j[\Algok]$ was incremented on line~\ref{al:SCS:mB6}, when validating some $\Algog'$, that was added in $\AlgoG_j$
after $p_j$ received a (first) message $M_l(\Algog'.\AlgoGK, \Algog'.\AlgoGT)$, with $\Algog'.\AlgoGK = \Algok$ and $\Algog'.\AlgoGT = \AlgoVC_j[\Algok]$.
Remark that, as FIFO reception is used, $p_{\Algok}$ sent message $M_{\Algok}(\Algok, \Algot)$ before $M_{\Algok}(\Algok, \AlgoVC_j[\Algok])$,
and all other processes only forward messages, $p_j$ received message $M_l(\Algok, \Algot)$ before $M_l(\Algok, \AlgoVC_j[\Algok])$, and at that time, 
$\Algot > \AlgoVC_j[\Algok]$, so the first case applies.

Now, we prove that $p_i$ will broadcast no other message with the same $\Algok$ and $\Algot$ later. 
If $i=\Algok$, the message would be sent on line~\ref{al:SCS:w3} or~\ref{al:SCS:mC3}, just after $\AlgoSC_i$ is incremented, 
which would lead to a different $\Algot$. Otherwise, the message would be sent on line~\ref{al:SCS:mA6}, which would mean 
the condition of line~\ref{al:SCS:mA2} is false. As $p_i$ broadcast a first message, a corresponding $\Algog$ was present in $Algog_i$,
deleted on line~\ref{al:SCS:mB4}, which would make the condition of line~\ref{al:SCS:mA1} to be false.
\end{proof}

\begin{lemma} \label{lemma:safety}
Let $p_i$, $p_j$ be two processes and $t_i$, $t_j$ be two time instants, and let us denote by $\AlgoVC_i^{t_i}$ (resp. $\AlgoVC_j^{t_j}$) the value of $\AlgoVC_i$ (resp. $\AlgoVC_j$) at time $t_i$ (resp. $t_j$). We have either, for all $k$, $\AlgoVC_i^{t_i}[k] \le \AlgoVC_j^{t_j}[k]$ or for all $k$, $\AlgoVC_j^{t_j}[k] \le \AlgoVC_i^{t_i}[k]$.
\end{lemma}

\begin{proof}
Let $p_i$, $p_j$ be two processes and $t_i$, $t_j$ be two instants. Let us suppose (by contradiction) that 
there exist $k$ and $k'$ such that $\AlgoVC_j^{t_j}[k] < \AlgoVC_{i}^{t_i}[k]$ and $\AlgoVC_i^{t_i}[k'] < \AlgoVC_{j}^{t_j}[k']$.

As $\AlgoVC_i$ is only updated on line \ref{al:SCS:mB6}, at some time $t_i^k \le t_i$, there was $\Algog_i^k\in G'$ with 
$\Algog_i^k.\AlgoGK = k$ and $\Algog_i^k.t = \AlgoVC_i^{t_i}[k]$. According to line \ref{al:SCS:mB1}, we have 
$|\{l : \Algog_i^k.\AlgoGCL[l] < \infty\}| > \frac{n}{2}$ and according to lines \ref{al:SCS:mA3} and \ref{al:SCS:mA8},
each finite field $\Algog_i^k.\AlgoGCL[l]$ corresponds to the reception of a message $M_l(k, \AlgoVC_i^{t_i}[k])$.
Similarly, process $p_j$ received messages $M_l(k', \AlgoVC_j^{t_j}[k'])$ from more than $\frac{n}{2}$ processes.
Since the number of processes is $n$, the intersection of these two sets of processes is not empty.

Let $p_c$ be a process that belongs to both sets, i.e. $p_c$ broadcast messages $M_c(k, \AlgoVC_i^{t_i}[k])$ and
$M_c(k', \AlgoVC_j^{t_j}[k'])$. Process $p_c$ sent these two messages in a given order, let us say
$M_c(k', \AlgoVC_j^{t_j}[k'])$ before $M_c(k, \AlgoVC_i^{t_i}[k])$ (the other case is symmetric).
As $\AlgoSC_c$ is never decremented and it is incremented before all sendings, $CL_c(k', \AlgoVC_j^{t_j}[k']) < CL_c(k, \AlgoVC_i^{t_i}[k])$. 
Moreover, as the protocol uses FIFO ordering, $p_i$ received the two messages in the same order.

According to line~\ref{al:SCS:mB6}, $\AlgoVC_i$ can only increase, so $\AlgoVC_i^{t'_i}[k'] \le \AlgoVC_i^{t_i}[k']$ and $\AlgoVC_i^{t_i}[k'] < \AlgoVC_j^{t_j}[k']$. It means that the condition on line~\ref{al:SCS:mA1} was true when 
$p_i$ received $M_c(k', \AlgoVC_j^{t_j}[k'])$. Then, after the execution of the condition starting on line~\ref{al:SCS:mA2} and whatever the result of this condition,
there was a $\Algog_i^{k'}\in \AlgoG_i$ with $\Algog_i^{k'}.\AlgoGK = k'$, $\Algog_i^{k'}.\AlgoGT = \AlgoVC_j^{t_j}[k']$ and $\Algog_i^{k'}.\AlgoGCL[c] = CL_c(k', \AlgoVC_j^{t_j}[k'])$.

At time $t_i$, if $g_i^{k'}\not\in G_i$, it was removed on line~\ref{al:SCS:mB4}, which means $\AlgoVC_i^{t}[k']\ge \Algog_i^{k'}.\AlgoGT = \AlgoVC_j^{t_j}[k']$
by lines~\ref{al:SCS:mB5} and~\ref{al:SCS:mB6}, which is absurd by our hypothesis. Otherwise, after line~\ref{al:SCS:mB2} was executed at time $t_i^k$,
we have $\Algog_i^{k}\in G'$ and $\Algog_i^{k'}\not\in G'$, which is impossible as $\Algog_i^{k'}.\AlgoGCL[c]\le \Algog_i^{k}.GCL[c]$.

This is a contradiction. Therefore $\AlgoVC_i^{t_i} \le \AlgoVC_j^{t_j}$ or $\AlgoVC_j^{t_j} \le \AlgoVC_i^{t_i}$.
\end{proof}

\begin{lemma} \label{lemma:liveness}
If a message $\AlgoM(\Algov,i,\Algot,\Algot)$ is sent by a correct process $p_i$, then beyond some time $t'$, 
for each correct process $p_j$, $\AlgoVC_j^{t'}[i] \ge \Algot$.
\end{lemma}

\begin{proof}
Let us suppose a message $M_i(i,t)$ is sent by a correct process $p_i$. 

Let us suppose (by contradiction) that there exists a process $p_j$ such that the pair $(i, t)$ has an infinity of
predecessors according to $\rightarrow_j^\star$.
As the number of processes is finite, an infinity of these predecessors correspond to the same process, let us say $(k, t_l)_{l\in \mathbb{N}}$.
As $p_j$ is correct, $p_k$ eventually receives message $M_j(i,t)$, which means an infinity of messages $m_k(k,t_l)$
were sent after $p_k$ receives message $m_j(i,t)$, and for all of them, $(k,t_l)\rightarrow_i^\star (i,t)\rightarrow_i (k,t_l)$.
Therefore, there exists a sequence $(k_1, t'_1) \rightarrow_i (k_2, t'_2) \rightarrow_i \dots \rightarrow_i (k_m, t'_m)$ 
with $k_1 = k_m = k$ and $t'_m > t'_1$. Two cases are possible for $(k_2, t'_2)$:
\begin{itemize}
\item If $p_k$ received a message $M_x(k_2, t'_2)$ (from any $p_x$) before it sent $M_k(k, t'_1)$,
  then $p_k$ also send $M_k(k_2, t'_2)$ before it sent $M_k(k, t'_1)$, and all processes 
  received these messages in the same order (and possibly a message $M_x(k_2, t'_2)$ even before from another process),
  which is in contradiction with the fact that $(k, t'_1) \rightarrow_i (k_2, t'_2)$.
\item Otherwise, there is an index $l$ such that process $p_{k}$ received a message $M_x(k_{l'}, t'_{l'})$ (from any $p_x$)
  for all $l'>l$ but not for $l'=l$, before it sent message $M_k(k, t'_1)$. Whether it finally sends it on line \ref{al:SCS:w3} or line \ref{al:SCS:mC3},
  there was no $\Algog\in \AlgoG_i$ corresponding to $(k_m, t'_m)$ so, by lines \ref{al:SCS:mB1}-\ref{al:SCS:mB2}, $p_{k}$ received messages 
  $M_x(k_{l'}, t'_{l'})$ for all $l'>l$, from a majority of processes $p_x$, and all of them sent $M_x(k_{l'}, t'_{l'})$ before $M_x(k_{l}, t'_{l})$.
  As $(k_{l}, t'_{l}) \rightarrow_i (k_{l+1}, t'_{l+1})$ and FIFO reception is used, a majority of processes sent 
  $M_x(k_{l}, t'_{l})$ before $M_x(k_{l'}, t'_{l'})$. This is impossible as two majorities always have a non-empty intersection.
  Therefore, this case is also impossible.
\end{itemize}
Finally, for all correct processes $p_j$, there exists a finite number of pairs $(k, t')$ such that $(k, t') \rightarrow_j (i, t)$. 
As $p_j$ is correct, according to Lemma~\ref{lemma:broadcast}, $p_j$ will eventually receive a message $M_x(k,t')$ for all of them
from all correct processes, which are in majority.
At the last message, on line \ref{al:SCS:mB5}, $G'$ will contain a $\Algog$ with $\Algog.\AlgoGK=i$ and $\Algog.\AlgoGT=t$
and after it executed line \ref{al:SCS:mB6}, it will have $\AlgoVC_j[i] \ge t$. As $\AlgoVC_j[i]$ can only grow and what precedes is true for all $j$,
eventually it will be true for all correct processes.
\end{proof}

Finally, given Lemmas~\ref{lemma:safety} and~\ref{lemma:liveness}, it is possible to prove that Algorithm \ref{algo:SCS} implements a sequentially consistent snapshot memory (Proposition~\ref{prop:correct}). The idea is to order snapshot operations according to the order given by Lemma~\ref{lemma:safety} on
the value of $\AlgoVC_i$ when they were made and to insert the update operations at the position where 
$\AlgoVC_i$ changes because they are validated. It is possible to complete this order into a linearization 
order, thanks to Lemma~\ref{lemma:liveness}, and to show that the execution of all the operations in that order
respects the sequential specification of the snapshot memory data structure.

\pagebreak
\begin{proposition}\label{prop:correct}
    All histories admitted by Algorithm \ref{algo:SCS} are sequentially consistent.
\end{proposition}

\begin{proof}
    Let $H$ be a history admitted by Algorithm \ref{algo:SCS}. For each operation $op$, let us define $op.clock$ as follows:
    \begin{itemize}
    \item If $op$ is a snapshot operation done on process $p_i$, $op.clock$ is the value of $\AlgoVC_i$ when $p_i$ executes line~\ref{al:SCS:r2}.
    \item If $op$ is an update operation done on process $p_i$, let us remark that the call to $op$ is followed by the sending of a message 
        $\AlgoM(v, i, cl_i, cl_i)$, either directly on line~\ref{al:SCS:w3} or later on line~\ref{al:SCS:mC3} as lemma~\ref{lemma:liveness} prevents
        the condition of line~\ref{al:SCS:mC1} to remain false forever (in this case, the value $v$ may be more recent from the one written in $op$).
        Let us consider the clock $cl_i$ of the first such message sent by $p_i$. We pose $op.clock$ as the smallest value taken by variable
        $\AlgoVC_j$ for any $j$ (according to the total order given by lemma~\ref{lemma:safety}) such that $op_i\le op.clock[i]$ 
        (such a clock exists according to lemma~\ref{lemma:liveness}).
    \end{itemize}
    Let $\lessdot$ be any total order on all the operations, that contains the process order,
    and such that all operation has a finite past according to $\lessdot$ ($\lessdot$ is only used to break ties).
    We define the relation $\leq$ on all operations of $H$ by $op\leq op'$ if
    \begin{enumerate}
    \item $op.clock < op'.clock$, or
    \item $op.clock = op'.clock$, $op$ is an update operation and $op'$ is a snapshot operation, or
    \item $op.clock = op'.clock$, $op$ and $op'$ are either two snapshot or two update operations and $op\lessdot op'$.
    \end{enumerate}
    Let us prove that $\leq$ is a total order.
    \begin{description}
        \item[reflexivity:] for all $op$, the third point in the definition is respected, as $\lessdot$ is a total order.
        \item[antisymmetry:] let $op, op'$ be two operations such that $op\leq op'\leq op$. We have $op.clock = op'.clock$, $op$ and $op'$ are either two snapshot or 
            two update operations and, as $\lessdot$ is antisymmetric, $op=op'$.
          \item[transitivity:] let $op, op', op''$ be three operations such that $op\leq op'\leq op''$.
            If $op.clock \leq op'.clock$ or $op'.clock \leq op''.clock$, then $op.clock \leq op''.clock$.
            Otherwise, $op.clock = op'.clock = op''.clock$. If the three operations are all update or all snapshot operations, $op.clock \lessdot op''.clock$ so $op.clock \leq op''.clock$.
            Otherwise, $op$ is an update and $op'$ is a snapshot so $op.clock \leq op''.clock$.
        \item[total:] let $op, op'$ be two operations. If $op.clock\neq op'.clock$, they are ordered according to lemma~\ref{lemma:safety}. Otherwise, they are ordered
            by one of the last two points.
    \end{description}

    Let us prove that $\leq$ contains the process order. Let $op$ and $op'$ be two operations that occurred on the same process $p_i$, on which $op$ preceded $op'$.
    According to lemma~\ref{lemma:safety}, $op.clock$ and $op'.clock$ are ordered. 
    \begin{itemize}
    \item If $op.clock < op'.clock$ then $op\leq op'$. 
    \item Let us suppose $op.clock = op'.clock$. It is impossible that $op$ is a read operation and $op'$ is an update operation: 
        as $\AlgoSC_i$ is always increased before $p_i$ sends a message, $op.clock[i] < op'.clock[i]$.
        If $op$ is an update operation and $op'$ is a snapshot operation, then $op\leq op'$. In the other cases, $op\lessdot op'$ so $op\leq op'$.
    \item We now prove case $op.clock > op'.clock$ cannot happen. As above, it is impossible that $o$ is a read operation and $o'$ is an update operation. 
        It is also impossible that $op$ and $op'$ are two read or two update operations because $\AlgoVC_i$ can only grow. 
        Finally, if $op$ is an update operation and $op'$ is a snapshot operation, $op.clock[i] \le op'.clock[i]$ thanks to line~\ref{al:SCS:r1}, 
        and by definition of $op.clock$ for update operations, $op.clock \le op'.clock$.
    \end{itemize}

    Let us prove that all operations have a finite past according to $\leq$. Let $op$ be an operation of the history. Let us first remark that, 
    for each process $p_i$, $op.clock[i]$ corresponds to a message $M_i(i, op.clock[i])$.
    According to lemma~\ref{lemma:liveness}, eventually, for all processes $p_i$, $\AlgoVC_i \ge op.clock$.
    Only a finite number of operations have been done before that, therefore $\{op': op'.clock < op.clock\}$ is finite. 
    Moreover, all the updates $op'$ with $op'.clock \le op.clock$ are done before that time, so there is a finite number of them. 
    If $op$ is an update operation, then its antecedents $op'$ verify either $op'.clock < op.clock$ or $op'.clock = op.clock$ and $op'$ is a write operation.
    In both cases, there is a finite number of them. If $op$ is a snapshot operation, its antecedents $op'$ verify either 
    (1) $op'.clock < op.clock$, (2) $op'.clock = op.clock$ and $op'$ is an update operation or (3) $op'.clock = op.clock$ and $op'$ is a snapshot operation. 
    Cases (1) and (2) are similar as above, and antecedents that verify case (3) also are its antecedents by $\lessdot$ so there is a finite number of them. 
    Finally, in all cases, $op$ has a finite number of antecedents. 
    
    Let us prove that the execution of all the operations in the order $\leq$ respects the sequential specification of memory. 
    Let $op$ be a snapshot operation invoked by process $p_i$ and let $p_j$ be a process. According to line~\ref{al:SCS:mB6}, 
    the value of $\AlgoX_i[j]$ corresponds to the value contained in a message $M_j(j, op.clock[j])$. Let $op'$ be the last 
    update operation invoked by process $p_j$ before it sent this message. Whether the message was sent on line \ref{al:SCS:w3} or \ref{al:SCS:mC3}, 
    $\AlgoX_i[j]$ is the value written by $op'$. Moreover, $op'.clock \le op.clock$ so $op'\leq o'$ and for all update operations $op''$ done by 
    process $p_j$ after $op'$, $op.clock < op''.clock$ so $op\leq op''$. All in all, $op$ returns the last values written on each register,
    according to the order $\leq$.
    
    Finally, $\leq$ defines a linearization of all the events of the history that respects the sequential specification of the shared object. 
    Therefore, $H$ is sequentially consistent.
\end{proof}

%%%%%%%%%%%%%%%%%%%%
\subsection{Complexity}\label{sec:complexity}
%%%%%%%%%%%%%%%%%%%%

\begin{figure}[t]
    \centering
    \scalebox{0.9}{
    \begin{tikzpicture}

      \draw (-0.5,0) -- (17.3,0) ;
      \draw (-0.5,0.5) -- (17.3,0.5) ;
      \draw (-0.5,1) -- (17.3,1) ;
      \draw (-0.5,1.5) -- (17.3,1.5) ;
      \draw (2.6,2.4) -- (17.3,2.4) ;

      \draw (-0.5,0) -- (-0.5,1.5) ;

      \draw (1,1.25) node{ABD \cite{attiya1995sharing}} ;
      \draw (1,0.75) node{ABD + AR \cite{attiya1995sharing,AttiyaR98}} ;
      \draw (1,0.25) node{Algorithm \ref{algo:SCS}} ;

      \draw (2.5,0) -- (2.5,1.5) ;
      \draw (2.6,0) -- (2.6,2.4) ;

      \draw (4.15,2.15) node{Read} ;

      \draw (3.5,1.75) node{\# messages} ;
      \draw (3.5,1.25) node{$\mathcal{O}(n)$} ;
      \draw[black!50] (3.5,0.75) node{$\sim$} ;
      \draw (3.5,0.25) node{$0$} ;

      \draw (4.4,0) -- (4.4,1.5) ;

      \draw (5.05,1.75) node{latency} ;
      \draw (5.05,1.25) node{$4$} ;
      \draw[black!50] (5.05,0.75) node{$\sim$} ;
      \draw (5.05,0.25) node{$0$ --- $4$} ;

      \draw (5.7,0) -- (5.7,2.4) ;

      \draw (7.25,2.15) node{Write} ;

      \draw (6.6,1.75) node{\# messages} ;
      \draw (6.6,1.25) node{$\mathcal{O}(n)$} ;
      \draw[black!50] (6.6,0.75) node{$\sim$} ;
      \draw (6.6,0.25) node{$\mathcal{O}(n^2)$} ;

      \draw (7.5,0) -- (7.5,1.5) ;

      \draw (8.15,1.75) node{latency} ;
      \draw (8.15,1.25) node{$2$} ;
      \draw[black!50] (8.15,0.75) node{$\sim$} ;
      \draw (8.15,0.25) node{$0$} ;

      \draw (8.8,0) -- (8.8,2.4) ;
      \draw (8.9,0) -- (8.9,2.4) ;

      \draw (11,2.15) node{Snapshot} ;

      \draw (9.95,1.75) node{\# messages} ;
      \draw[black!50] (9.95,1.25) node{$\sim$} ;
      \draw (9.95,0.75) node{$\mathcal{O}\left(n^2\log n\right)$} ;
      \draw (9.95,0.25) node{$0$} ;

      \draw (11,0) -- (11,1.5) ;

      \draw (12.05,1.75) node{latency} ;
      \draw[black!50] (12.05,1.25) node{$\sim$} ;
      \draw (12.05,0.75) node{$\mathcal{O}\left(n\log(n) \right)$} ;
      \draw (12.05,0.25) node{$0$ --- $4$} ;

      \draw (13.1,0) -- (13.1,2.4) ;

      \draw (15.2,2.15) node{Update} ;

      \draw (14.15,1.75) node{\# messages} ;
      \draw[black!50] (14.15,1.25) node{$\sim$} ;
      \draw (14.15,0.75) node{$\mathcal{O}\left(n^2\log n\right)$} ;
      \draw (14.15,0.25) node{$\mathcal{O}(n^2)$} ;

      \draw (15.2,0) -- (15.2,1.5) ;

      \draw (16.25,1.75) node{latency} ;
      \draw[black!50] (16.25,1.25) node{$\sim$} ;
      \draw (16.25,0.75) node{$\mathcal{O}\left(n\log(n)\right)$} ;
      \draw (16.25,0.25) node{$0$} ;

      \draw (17.3,0) -- (17.3,2.4) ;

    \end{tikzpicture}
    }
    \caption{Complexity of several algorithms to implement a shared memory.}
    \label{fig:complexity}
\end{figure}

In this section, we analyze the algorithmic complexity of Algorithm~\ref{algo:SCS} in terms of the number of messages 
and latency for snapshot and update operations. Fig.~\ref{fig:complexity} sums up this complexity and compares it with 
the standard implementation of linearizable registers \cite{attiya1995sharing}, as well as with the construction of a snapshot object \cite{AttiyaR98}
implemented on top of registers.

In an asynchronous system as the one we consider, the latency $d$ and the uncertainty $u$ of the network can not be expressed by constants. 
We therefore measure the complexity as the length of the longest chain of causally related messages to expect before an operation can complete.
For example, if a process sends a message to another process and then waits for its answer, the complexity will be $2$.

According to Lemma~\ref{lemma:broadcast}, it is clear that each update operation generates at most $n^2$ messages. 
The time complexity of an update operation is $0$, as update operations return immediately. 
No message is sent for snapshot operations. Considering its latency, in the worst case, a snapshot operation
is called immediately after two update operations $a$ and $b$. In this case, the process must wait until its own message for $a$
is received by the other processes, then to receive their acknowledgements, and then the same two messages must be routed for $b$,
which leads to a complexity of $4$. However, in the case of two consecutive snapshots, or if enough time has elapsed between a snapshot and the last update,
the snapshot can also return immediately.

In comparison, the ABD simulation uses solely a linear number of messages per operation (reads as well as writes), but waiting is necessary for both 
kinds of operations. Even in the case of the read operation, our worst case corresponds to the latency of the ABD simulation. 
Moreover, our solution directly implements the snapshot operation. Implementing a snapshot operation on top of a linearizable shared memory is 
actually more costly than just reading each register once. The AR implementation \cite{AttiyaR98}, that is (to our knowledge) 
the implementation of the snapshot that uses the least amount of operations on the registers, uses $\mathcal{O}(n\log n)$ operations on registers
to complete both a snapshot and an update operation. As each operation on memory 
requires $\mathcal{O}(n)$ messages and has a latency of $\mathcal{O}(1)$, our approach leads to a better performance in all cases. 

Algorithm~\ref{algo:SCS}, like~\cite{attiya1995sharing}, uses unbounded integer values to timestamp messages. 
Therefore, the complexity of an operation depends on the number $m$ of operations executed before it, in the linear extension. 
All messages sent by Algorithm~\ref{algo:SCS} have a size of 
$\mathcal{O}\left(log(n m)\right)$. In comparison, ABD uses messages of size $\mathcal{O}\left(log(m)\right)$ but implements only one register, 
so it would also require messages of size $\mathcal{O}\left(log(n m)\right)$ to implement an array of $n$ registers.

Considering the use of local memory, due to asynchrony, it is possible in some cases that $\AlgoG_i$ contains an entry $\Algog$
for each value previously written. In that case, the space occupied by $\AlgoG_i$ may grow up to $\mathcal{O}(m n\log m)$.
Remark however that, according to Lemma~\ref{lemma:safety}, an entry $\Algog$ is eventually removed from $\AlgoG_i$ 
(in a synchronous system, after $2$ time units if $\Algog.\AlgoGK = i$ or $1$ time unit if $\Algog.\AlgoGK \neq i$). 
Therefore, this maximal bound is not likely to happen. Moreover, if all processes stop writing (which is the case in the round based model
we discussed in Section~\ref{sec:round}), then eventually $\AlgoG_i$ becomes empty and the space occupied by the algorithm drops down to
$\mathcal{O}(n\log m)$, which is comparable to ABD. In comparison, the AR implementation keeps a tree containing past values from all registers, 
in each register, which leads to a much higher size of messages and local memory.

%%%%%%%%%%%%%%%%%%%%
\vspace{-1mm}
\section{Conclusion}\label{sec:conclusion}
%%%%%%%%%%%%%%%%%%%%

In this paper, we investigated the advantages of focusing on sequential consistency. Because of its non
composability, sequential consistency has received little focus so far. However, we show that in many applications, 
this limitation is not a problem. The first case concerns applications built on a layered architecture.
If one layer contains only one object, then it is impossible for objects built on top of it to determine if this
object is sequentially consistent or linearizable. The other example concerns round-based algorithms: if processes
access to one different sequentially consistent object in each round, then the overall history is also sequentially consistent.

Using sequentially consistent objects instead of their linearizable counterpart can be very profitable in terms of execution time of operations. Whereas waiting 
is necessary for both read and write operations when implementing linearizable memory,
we presented an algorithm in which waiting is only required for read operations when they follow directly a write operation.
This extends the result of Attiya and Welch (that only concerns synchronous failure-free systems) to asynchronous systems with crashes. Moreover, the proposed algorithm implements a sequentially consistent snapshot memory 
for the same cost, which results in a better message and time comlexity, for both kinds of operations, 
than the best known implementation of a snapshot memory.

Exhibiting such an algorithm is not an easy task for two reasons. First, as write operations are wait-free, a process may
write before its previous write has been acknowledged by other processes, which leads to ``concurrent'' write operations by the same process. 
Second, proving that an implementation is sequentially consistent is more difficult than proving it is linearizable since 
the condition on real time that must be respected by linearizability highly reduces the number of linear extensions that need to be considered.

\section{Acknowledgments}\label{sec:acknowledgments}

This work has been partially supported by the Franco-German ANR project DISCMAT under grant agreement ANR-14-CE35-0010-01. 
The project is devoted to connections between mathematics and distributed computing.

\bibliographystyle{plain}
\bibliography{biblio}

\begin{thebibliography}{10}

\bibitem{Afek93}
Yehuda Afek, Hagit Attiya, Danny Dolev, Eli Gafni, Michael Merritt, and Nir
  Shavit.
\newblock Atomic snapshots of shared memory.
\newblock {\em J. {ACM}}, 40(4):873--890, 1993.

\bibitem{attiya1995sharing}
Hagit Attiya, Amotz Bar-Noy, and Danny Dolev.
\newblock {Sharing memory robustly in message-passing systems}.
\newblock {\em Journal of the ACM (JACM)}, 42(1):124--142, 1995.

\bibitem{AHR95}
Hagit Attiya, Maurice Herlihy, and Ophir Rachman.
\newblock Atomic snapshots using lattice agreement.
\newblock {\em Distributed Computing}, 8(3):121--132, 1995.

\bibitem{AttiyaR98}
Hagit Attiya and Ophir Rachman.
\newblock Atomic snapshots in o(n log n) operations.
\newblock {\em {SIAM} J. Comput.}, 27(2):319--340, 1998.

\bibitem{attiya1994sequential}
Hagit Attiya and Jennifer~L Welch.
\newblock {Sequential consistency versus linearizability}.
\newblock {\em ACM Transactions on Computer Systems (TOCS)}, 12(2):91--122,
  1994.

\bibitem{Awerbuch85}
Baruch Awerbuch.
\newblock Complexity of network synchronization.
\newblock {\em J. {ACM}}, 32(4):804--823, 1985.

\bibitem{birman1987reliable}
Kenneth~P Birman and Thomas~A Joseph.
\newblock {Reliable communication in the presence of failures}.
\newblock {\em ACM Transactions on Computer Systems (TOCS)}, 5(1):47--76, 1987.

\bibitem{BorowskyG92}
Elizabeth Borowsky and Eli Gafni.
\newblock Immediate atomic snapshots and fast renaming (extended abstract).
\newblock In {\em Proceedings of the Twelth Annual {ACM} Symposium on
  Principles of Distributed Computing, Ithaca, New York, USA, August 15-18,
  1993}, pages 41--51, 1993.

\bibitem{CS09}
Bernadette Charron{-}Bost and Andr{\'{e}} Schiper.
\newblock The heard-of model: computing in distributed systems with benign
  faults.
\newblock {\em Distributed Computing}, 22(1):49--71, 2009.

\bibitem{DLS88}
C.~Dwork, N.A. Lynch, and L.J. Stockmeyer.
\newblock Consensus in the presence of partial synchrony.
\newblock {\em J. {ACM}}, 35(2):288--323, 1988.

\bibitem{dwork1992time}
Cynthia Dwork, Maurice Herlihy, Serge~A Plotkin, and Orli Waarts.
\newblock Time-lapse snapshots.
\newblock In {\em Theory of Computing and Systems}, pages 154--170. Springer,
  1992.

\bibitem{G98}
E.~Gafni.
\newblock {\em Distributed Computing: a Glimmer of a Theory, in Handbook of
  Computer Science}.
\newblock CRC Press, 1998.

\bibitem{Gafni98}
Eli Gafni.
\newblock Round-by-round fault detectors: Unifying synchrony and asynchrony
  (extended abstract).
\newblock In {\em Proceedings of the Seventeenth Annual {ACM} Symposium on
  Principles of Distributed Computing, {PODC} '98, Puerto Vallarta, Mexico,
  June 28 - July 2, 1998}, pages 143--152, 1998.

\bibitem{goodman1991cache}
James~R Goodman.
\newblock {\em {Cache consistency and sequential consistency}}.
\newblock University of Wisconsin-Madison, Computer Sciences Department, 1991.

\bibitem{herlihy1990linearizability}
Maurice~P Herlihy and Jeannette~M Wing.
\newblock {Linearizability: A correctness condition for concurrent objects}.
\newblock {\em ACM Transactions on Programming Languages and Systems (TOPLAS)},
  12(3):463--492, 1990.

\bibitem{KirousisST94}
Lefteris~M. Kirousis, Paul~G. Spirakis, and Philippas Tsigas.
\newblock Reading many variables in one atomic operation: Solutions with linear
  or sublinear complexity.
\newblock {\em {IEEE} Trans. Parallel Distrib. Syst.}, 5(7):688--696, 1994.

\bibitem{lamport1979make}
Leslie Lamport.
\newblock {How to make a multiprocessor computer that correctly executes
  multiprocess programs}.
\newblock {\em Computers, IEEE Transactions on}, 100(9):690--691, 1979.

\bibitem{lipton1988pram}
Richard~J Lipton and Jonathan~S Sandberg.
\newblock {\em {PRAM: A scalable shared memory}}.
\newblock Princeton University, Department of Computer Science, 1988.

\bibitem{RST01}
Yaron Riany, Nir Shavit, and Dan Touitou.
\newblock Towards a practical snapshot algorithm.
\newblock {\em Theor. Comput. Sci.}, 269(1-2):163--201, 2001.

\end{thebibliography}

\end{document}